\documentclass[12pt]{amsart}

\usepackage[hmargin=0.8in,height=8.6in]{geometry}
\usepackage{amssymb,amsthm, times, enumerate}
\usepackage{delarray,verbatim}

\usepackage{ifpdf}
\ifpdf
\usepackage[pdftex]{graphicx}
 \else
\usepackage[dvips]{graphicx}
 \fi

\usepackage{ifpdf}
\usepackage{color}
\definecolor{webgreen}{rgb}{0,.5,0}
\definecolor{webbrown}{rgb}{.8,0,0}
\definecolor{emphcolor}{rgb}{0.95,0.95,0.95}

\usepackage{hyperref}
\hypersetup{%
          colorlinks=true,
          linkcolor=webbrown,
          filecolor=webbrown,
          citecolor=webgreen,
          breaklinks=true}
\ifpdf \hypersetup{pdftex,
            pdfstartview=FitH, 
            bookmarksopen=true,
            bookmarksnumbered=true
} \else \hypersetup{dvips} \fi

\numberwithin{equation}{section}

\newtheorem{theorem}{Theorem}[section]
\newtheorem{proposition}{Proposition}[section]

\newtheorem{remark}{Remark}[section]
\newtheorem{lemma}{Lemma}[section]

\newtheorem{assump}{Assumption}[section]

\numberwithin{remark}{section} \numberwithin{proposition}{section}
\numberwithin{corollary}{section}
\newcommand {\R}{\mathbb{R}}

\newcommand {\F}{\mathcal{F}}

\newcommand {\p}{\mathbb{P}}

\newcommand {\E}{\mathbb{E}}

\newcommand{\diff}{{\rm d}}

\newcommand{\lev}{L\'{e}vy }

\title[Optimal Capital Structure with Scale effects]{Optimal Capital Structure with Scale effects under spectrally negative L\'evy models$^*$}

\thanks{$*$\,This version: \today.  An earlier version of this paper was circulated as
``Toward a Generalization of the Leland-Toft Optimal Capital Structure Model".  The authors thank S.\ Boyarchenko, J.-C. Duan, M.\ Kijima, S.\ Kou, A.\ Kyprianou, K.\ Nishide, M.\ Nishihara, A.\ Novikov, J.\ Sekine, T.\ Suzuki as well as an anonymous referee and the editor for helpful suggestions and remarks.
K.\ Yamazaki is in part supported by MEXT KAKENHI Grant No.\ 22710143 and by JSPS KAKENHI Grant No.\  2271014.}

\author[B. A. Surya]{\,\,Budhi Arta Surya$^\dag$\,}\thanks{$\dag$\,School of Business and Management,
Bandung Institute of Technology,
Jalan Ganesha No.10, Bandung 40132, Indonesia. Email: \mbox{{\em
        budhi.surya@sbm-itb.ac.id}}}
\author[K. Yamazaki]{\,Kazutoshi Yamazaki$^\ddag$}\thanks{$\ddag$\, (corresponding author) Department of Mathematics,
Faculty of Engineering Science, Kansai University, 3-3-35 Yamate-cho, Suita-shi, Osaka 564-8680, Japan.  Email: \mbox{{\em
kyamazak@kansai-u.ac.jp}}. Phone: +81-6-6368-152.  }

\begin{document}
\maketitle
\begin{abstract}
The optimal capital structure model with endogenous bankruptcy was first studied by Leland \cite{Leland_1994} and Leland and Toft  \cite{Leland_Toft_1996}, and was later extended to the spectrally negative \lev model by Hilberink and Rogers \cite{ Hilberink_Rogers_2002} and Kyprianou and Surya \cite{ Kyprianou_Surya_2007}.  
%
This paper  incorporates scale effects by allowing the values of bankruptcy costs and tax benefits to be dependent on the firm's asset value.  
By using the fluctuation identities for the spectrally negative \lev process, we obtain a candidate bankruptcy level as well as a sufficient condition for optimality. The optimality holds in particular when, monotonically in the asset value, the value of tax benefits is increasing, the loss amount at bankruptcy is increasing, and its proportion relative to the asset value is decreasing. The solution admits a semi-explicit form in terms of the scale function.  A series of numerical studies are given to analyze the impacts of scale effects on the  bankruptcy strategy and the optimal capital structure.
\end{abstract}
{\noindent \small{\textbf{Keywords:}\,  Credit risk, spectrally negative \lev processes, scale functions, optimal stopping }}

\section{Introduction} \label{section_introduction}
We revisit the Leland-Toft optimal capital structure model \cite{Leland_Toft_1996} with endogenous bankruptcy.  A firm is partly financed by debt of equal seniority that is continuously retired and reissued so that its maturity profile is kept constant through time.  It distributes a continuous stream of coupon payments to bondholders, on which the firm receives tax benefits. The model assumes that the shareholders have the right to determine the time of bankruptcy so as as to maximize the firm's equity value subject to the limited liability constraint.  This gives a framework for obtaining the optimal capital structure that solves the tradeoff between minimizing bankruptcy costs and maximizing tax benefits.



This model was first studied by \cite{Leland_1994,Leland_Toft_1996} where they assumed a geometric Brownian motion for the firm's asset value, and was later extended to a \lev model, among others, by \cite{Chen_Kou_2009,Dao_Jeanblanc_2012, Hilberink_Rogers_2002,Kyprianou_Surya_2007,  Courtois_2008}.  By introducing jumps,  it allows the value of bankruptcy costs to be stochastic, and more importantly resolves the contradictory conclusion under the continuous diffusion model that the credit spreads go to zero as the maturity decreases to zero.  The problem reduces to a non-standard optimal stopping problem, and the optimal bankruptcy level can be obtained via the continuous/smooth fit principle.  It was solved, in particular, for the double exponential jump diffusion process \cite{Chen_Kou_2009}, for stable processes \cite{Courtois_2008}, and for a general spectrally negative \lev process \cite{ Hilberink_Rogers_2002,  Kyprianou_Surya_2007}.

%
%
%
%
%
%
%
%

%

In this paper, we add more dynamics to the \lev model by incorporating the ``scale effect" or the inhomogeneity with respect to the firm size.
Despite the fascinating contributions of the aforementioned papers, there are several assumptions on the bankruptcy costs and tax benefits, which are rather artificially imposed to derive explicit/analytical solutions.   This paper is aimed to relax these assumptions.  Here we review the assumptions required in \cite{Leland_Toft_1996} and its extensions and also give empirical evidence that motivates us to relax these assumptions.


\textbf{Bankruptcy costs:} \label{section_introduction_bankruptcy_cost}
Regarding the bankruptcy costs, it is commonly assumed that their value is proportional to the firm's asset value at the time of bankruptcy.   However, this is empirically rejected, and it is widely accepted that the amount of bankruptcy costs as a ratio of the asset value is dependent on the size of the firm. Following \cite{Ang_et_al_1982,Warner_1976}, we call this the \emph{scale effect}.

One of the earliest empirical studies was conducted by Warner  \cite{Warner_1976} where he observed this effect based on 
the direct bankruptcy costs of the  11 railroad companies that liquidated between 1933 and 1955.  In Figure \ref{figure_bankruptcy}, we apply linear regressions to Figures 1 and 2 of   \cite{Warner_1976}, which plot, respectively, the bankruptcy costs and their percentages relative to the market values of the firms.  Although this is not a careful analysis of these data, it is still reasonable to conclude that, while the cost tends to increase,  its proportion relative to the firm value tends to decrease.  Ang et al.\  \cite{Ang_et_al_1982} also confirmed this, observing via regression methods the strict concavity of the bankruptcy cost function based on the bankruptcy data between 1963 and 1978 in the Western District of Oklahoma.  The scale effect is also consistent with the more recent and comprehensive empirical studies conducted by \cite{Bris_2006} under corporate bankruptcy data in Arizona and New York between 1995 and 2001. 

\begin{figure}[tbp]
\begin{center}
\begin{minipage}{1.0\textwidth}
\centering
\begin{tabular}{cc}
\includegraphics[scale=0.35]{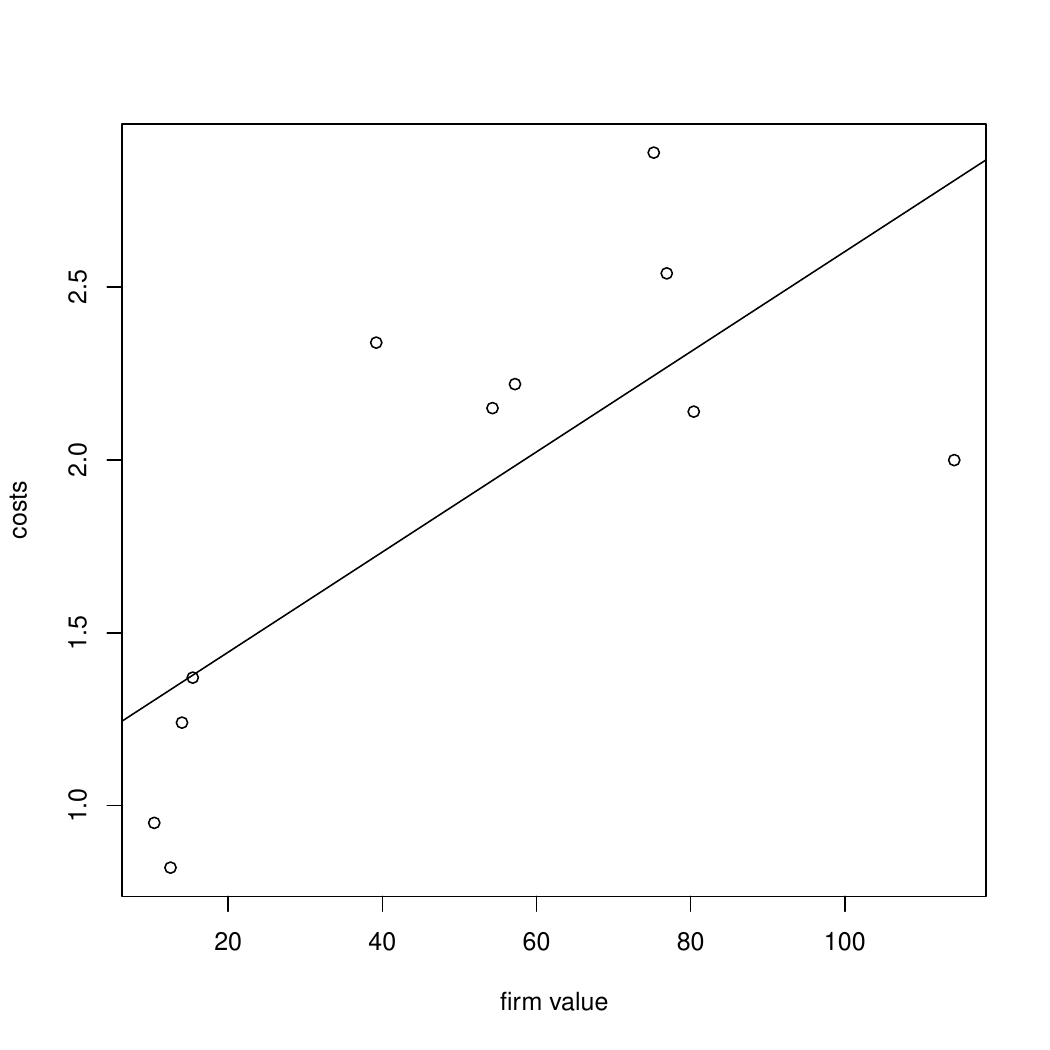}  & \includegraphics[scale=0.35]{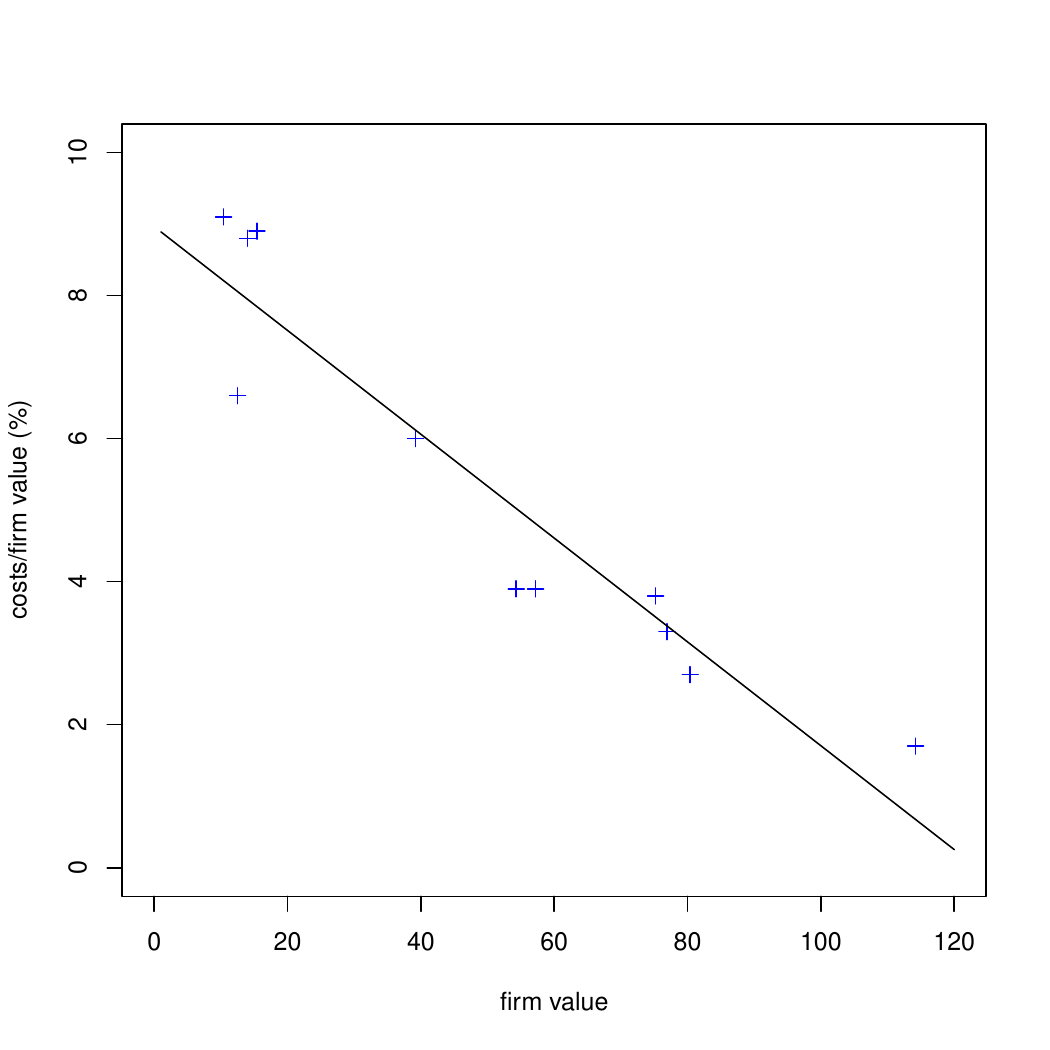} 
\end{tabular}
\end{minipage}
\caption{\small{(left) Bankruptcy costs and (right) percentage of bankruptcy costs with respect to the firm value.}} \label{figure_bankruptcy}
\end{center}
\end{figure}

These empirical results only consider the direct costs, typically consisting of legal and administrative costs.  On the other hand, indirect costs are in general much harder to compute, and the estimation tends to vary heavily on the nature of the data and also on the approximation methods.  
While the exact amount of indirect bankruptcy costs is hard to estimate, the existence of the scale effect is easy to be validated.
For example, as discussed in \cite{Franks_Torous,Thorburn_2000}, the time spent in bankruptcy can be used as a proxy for indirect bankruptcy costs. Bris et al.\ \cite{Bris_2006} analyzed the mean time in bankruptcy with respect to the firm size, both for Chapter 7 liquidations and Chapter 11 reorganizations (see \cite{Francois_2004} regarding the impact of Chapter 11 in comparison to Chapter 7.)  They observed that larger bankruptcies take longer to resolve, but the relative increase is small, implying the scale effect in indirect bankruptcy costs.


Due to the difficulty of estimating the exact bankruptcy costs, it is not possible to draw a general conclusion.
However, it is safe to state that the proportionality assumption of the bankruptcy costs relative to the asset value is an oversimplification, and one needs to incorporate the scale effect for  realistic models. Furthermore, the characteristics of the bankruptcy costs vary heavily across industries, regions and time periods, and this further motivates us to pursue more flexible formulations.

\textbf{Tax benefits:} \label{section_introduction_tax}
As for the tax benefits, empirical evidence also suggests an effect similar to the scale effect as in the case of bankruptcy costs.  In other words, the marginal effect of  taxable income reduction per dollar is dependent on the amount of the total taxable income that changes dynamically over time.

In the original Leland model \cite{Leland_1994} and the majority of its extensions, 
this dependency is ignored and a linear tax schedule, or equivalently a constant tax rate, is assumed.
Because the face value and the coupon rate are fixed constant, the rate of tax benefits also stays constant over time until bankruptcy.
As is clear from the fact that the tax exemption is not necessarily effective for example when the coupon payments exceed the profits, 
this simplification overestimates the  tax benefits and result in  optimal leverage that tends to be excessively higher than  what is observed in reality (see, e.g., \cite{Berens_1995}).

In order to avoid this, tax cutoff is often applied.  Namely,  the tax rate is assumed to be a step function that is a full corporate tax rate when the asset value (as a proxy to the taxable income value) is above a pre-determined level and it is zero otherwise.  For spatially-homogeneous processes such as Brownian motion and \lev processes, this can be handled relatively easily and can potentially improve the solution if the cutoff level is chosen appropriately.  
However, this is what \cite{ Hilberink_Rogers_2002} calls ``an idealisation"  and even a rough estimate of these cutoff levels is hard to compute.

While the tax structure is an integral factor in determining the capital structure, the estimation of the tax function, or the tax amount with respect to the taxable income, is already very difficult.  Roughly speaking, income is taxable but loss is not.  This implies
a piecewise-linear convex tax function with zero values on the negative part and a straight line on the positive part.  However, there are many factors that deform its shape.  For example, for a firm facing tax progressivity, its effective marginal tax rate (or the derivative of its tax function) is monotonically increasing as opposed to being a constant.  Graham and Smith \cite{Graham_Smith_1999}, among others, examined the effects of statutory progressivity, net operation loss carrybacks/forwards, investment tax credits, the alternative minimum tax and uncertainty in taxable income.  After these are taken into consideration, the kink at zero of the (piecewise-linear) tax function is smoothed out and results in another convex function that is strictly convex for small taxable income while it gets closer to linear as the taxable income increases (see Figures 1 and 2 of \cite{Graham_Smith_1999}).  This implies that the marginal effect of taxable income reduction (or the slope of the tax function) is an increasing function that converges to the full tax rate at infinity.   The effect of this convexity in the determination of the capital structure has been examined by, for example, \cite{Sarkar_2008}. It is shown in their framework that the effects of tax convexity brings the optimal bankruptcy boundary higher and reduces the optimal leverage.

Although there are a number of papers rectifying the overestimation of the tax benefits, most of them employ the step-function tax rate (or piecewise linear tax function) similarly to the cutoff approach described above.   However, this does not fully reflect the strict convexity for small taxable income. The determination of tax benefits is heavily affected especially for firms  with small and/or volatile taxable income. Furthermore, the tax-code varies across countries and across  industries.  For these reasons, we need a more flexible formulation of tax benefits that captures the dynamics and the tax function convexity.

\vspace{0.23cm}



In this paper, we resolve these inflexibilities of the existing models by expressing the values of bankruptcy costs and tax benefits as functions of the firm's asset value.  This generalization encompasses the existing models and adds flexibility in describing a more realistic capital structure.  We focus on the spectrally negative \lev model considered by \cite{ Hilberink_Rogers_2002, Kyprianou_Surya_2007} where the firm's asset value is driven by a general \lev process with only negative jumps.   We obtain a sufficient condition for optimality of a bankruptcy strategy and show its optimality for example when, monotonically in the asset value, the value of tax benefits is increasing, the loss amount at bankruptcy is increasing, and its proportion relative to the asset value is decreasing.  This condition is consistent with the empirical studies reviewed above on the bankruptcy costs and tax benefits.

Our main challenge is to obtain the endogenous bankruptcy strategy by the equity holder.  While it is fairly well-known that the continuous and/or smooth fit can be applied for the exponential \lev model as in \cite{Chen_Kou_2009, Hilberink_Rogers_2002, Kyprianou_Surya_2007}, their results rely particularly on the form of the bankruptcy cost as an exponential function of the \lev process; once the function is generalized, it is no longer clear if the same principle is directly applicable.  However, thanks to the recent advances in the spectrally negative \lev process and its fluctuation theories (e.g.\ \cite{Bertoin_1996, Kyprianou_2006}), the equity value to be maximized can then be written in terms of the scale function and, as we shall see, the optimal bankruptcy level can be derived in a general setting.

The spectrally negative \lev model has been drawing much attention recently in mathematical finance and insurance as generalizations of the classical Black-Scholes  and  Cram\'er-Lundberg models.  A  number of authors have succeeded in extending the classical results to the spectrally negative \lev model by way of scale functions.  We refer the reader to \cite{Baurdoux2008,Baurdoux2009}  for stochastic games,  \cite{Avram_et_al_2007, Bayraktar_2012, Bayraktar_2013, Kyprianou_Palmowski_2007, Loeffen_2008}  for the optimal dividend problem, \cite{alili-kyp, avram-et-al-2004} for American and Russian options,  and \cite{Leung_Yamazaki_2011, Egami-Yamazaki-2010-1,  Kyprianou_Surya_2007, Leung_Yamazaki_2010} for credit risk.   In particular, Egami and Yamazaki \cite{Egami-Yamazaki-2011} considered a general optimal stopping problem for spectrally negative \lev processes and obtained the first-order condition for maximization over threshold strategies; the results are also confirmed numerically by \cite{Yamazaki_2012} in a multiple stopping setting.  Using their results, we show under suitable assumptions that the derivative of  the equity value with respect to the bankruptcy level is monotone, and hence the optimal bankruptcy level can be obtained.

%

The solutions to our generalized model admit semi-explicit expressions written in terms of the scale function, which has analytical forms in certain cases \cite{Egami-Yamazaki-2011, Hubalek_Kyprianou_2009, Kyprianou_2008, Kyprianou_Surya_2007} and can be approximated generally using, e.g., \cite{Egami_Yamazaki_2010_2,Surya_2008}.  In order to illustrate the implementation side, we give an example based on a mixture of Brownian motion and a compound Poisson process with i.i.d.\ exponential jumps.  We compute the optimal bankruptcy levels and the corresponding equity/debt/firm values as well as the optimal leverage ratios as solutions to the \emph{two-stage problem} considered in \cite{Chen_Kou_2009,Leland_1994, Leland_Toft_1996}.  We conduct a sequence of numerical experiments to analyze the impacts of the scale effects on the bankruptcy strategy and the capital structure.


The rest of the paper is organized as follows.  In Section \ref{section_problem},  we review the existing \lev model and introduce our generalized model.   In Section \ref{section_solution}, we focus on the spectrally negative \lev model and obtain a sufficient condition for optimality.
Section \ref{section_example} shows examples that satisfy the sufficient condition.  We give numerical results in Section \ref{section_numerics}.  All proofs are deferred to the Appendix. Throughout the paper, unless otherwise stated, monotonicity is understood in the weak sense; ``increasing" and ``decreasing" are equivalent to ``nondecreasing" and ``nonincreasing", respectively.

%

%
%
%


\section{Problem Formulation} \label{section_problem}
In this section, we first review the existing model and then generalize it.  In particular, we adopt the formulation and notations by \cite{Hilberink_Rogers_2002, Kyprianou_Surya_2007} to a maximum extent. The formulation addressed here holds for any \lev model; in the next section, we focus on the spectrally negative \lev process and derive optimal solutions.

\subsection{The optimal capital structure  model by \cite{Hilberink_Rogers_2002, Kyprianou_Surya_2007}} \label{subsection_model_classic} 
Let $(\Omega, \F, \p)$ be a complete probability space hosting a  L\'{e}vy process $X=\{X_t;\, t\ge 0\}$.  The value of the \emph{firm's asset} is assumed to evolve according to an \emph{exponential \lev process} $V_t := e^{X_t}$, $t\geq0$.   Let $r > 0$ be the positive risk-free interest rate and $0 \leq \delta < r$ the total payout rate to the firm's investors.  
We assume that the market is arbitrage-free and that $\p$ is a risk neutral probability.  This requires $\{ e^{-(r-\delta) t} V_t ; t \geq 0 \}$ to be a $\p$-martingale.  We denote by  $\p_x$ the probability law and $\E_x$ the expectation under which $X_0=x$ (or equivalently $V_0 = e^x$).  

The firm is partly financed by debt with a constant debt profile; it issues new debt at a constant rate $p$ with maturity profile $\varphi(s) := m e^{-ms}$ for some given constants $p, m > 0$.  Namely,  in the time interval $(t, t+\diff t)$, it issues debt with face value $p \varphi(s) \diff t \diff s$ that matures in the time interval $(t+s, t+s+\diff s)$.
By this assumption, at time $0$, the face value of debt that matures in $(s, s+ \diff s)$ becomes
\begin{align}
\left[  \int_{-\infty}^0 p \varphi(s-u) \diff u \right] \diff s = p e^{-m s} \diff s, \label{mat_profile}
\end{align}
and the face value of all debt is a constant value,
\begin{align*}
P := \int_0^\infty p e^{-m s} \diff s = \frac p m.
\end{align*}

Suppose the bankruptcy is triggered at the first time $X$ goes below a given level $B \in \R$, or
\begin{align}
\tau_B^-  := \inf \left\{ t \geq 0: X_t \leq B \right\}, \quad B \in \R. \label{default_time}
\end{align}
The debt pays a constant coupon flow at a fixed rate $\hat{\rho} > 0$ and a constant fraction $0 \leq \hat{\eta} \leq 1$ of the asset value is lost at the bankruptcy time $\tau_B^-$;  the value of the debt with a unit face value  and maturity $t>0$ becomes 
\begin{align}
d (x; B, t) := \E_x \Big[ \int_0^{t \wedge \tau_B^- } e^{-rs} \hat{\rho} \diff s \Big] + \E_x \left[ e^{-rt} 1_{\{t < \tau_B^- \}}\right] + \frac 1 P \E_x \Big[ e^{-r\tau_B^- + X_{\tau_B^-} } \left(1-\hat{\eta}  \right) 1_{\{\tau_B^- < t \}}\Big]. \label{debt_constant_unit}
\end{align}
Here, the first term is the total value of the coupon payments accumulated until maturity or bankruptcy whichever comes first.  The second term is the value of the principle payment.  The last term corresponds to the $1/P$ fraction of the asset value that is distributed, in the event of bankruptcy, to the bondholder of a unit face value.
The \emph{total value of debt} becomes, by \eqref{mat_profile} and Fubini's theorem,
\begin{align*}
\mathcal{D} (x; B) &:= \int_0^\infty p e^{-m t} d (x; B, t)  \diff t \\
&= \E_x\Big[ \int_0^{\tau_B^-} e^{-(r+m)t} \left( P \hat{\rho}+ p \right) \diff t\Big] +  \E_x \Big[ e^{-(r+m) \tau_B^- + X_{\tau_B^-} } \left(1-\hat{\eta}  \right) 1_{\{\tau_B^- < \infty \}}\Big].
\end{align*}

Regarding the \emph{(market) value of the firm}, it is assumed that there is a corporate tax rate $\hat{\gamma}\in [0,1]$ and its (full) rebate on coupon payments is gained if and only if $V_t \geq v_T$ (or $X_t \geq \log v_T$) for some given cutoff level $v_T > 0$.  Based on the Modigliani-Miller theorem (see, e.g., \cite{Modilliani_Miller_1958, Modilliani_Miller_1963}), the firm value becomes
\begin{align*}
\mathcal{V}(x;B) &:= e^x + \E_x \Big[ \int_0^{\tau_B^-} e^{-rt} 1_{\{X_t \geq \log v_T \}} P \hat{\gamma} \hat{\rho} \diff t\Big] -  \hat{\eta}  \E_x \Big[ e^{-r \tau^-_B + X_{\tau_B^-}} 1_{\{\tau_B^- < \infty \}}   \Big],
\end{align*}
where each term corresponds to the current (unlevered) asset value, the total value of tax benefits and the value of loss at bankruptcy, respectively.

The problem is to pursue an \emph{optimal bankruptcy level} $B \in \R$ that maximizes the \emph{equity value},
\begin{align}
\mathcal{E}(x;B) :=\mathcal{V}(x;B) - \mathcal{D} (x; B), \quad x > B, \label{equity}
\end{align}
 subject to the \emph{limited liability constraint},
\begin{align}
\mathcal{E}(x;B) \geq 0, \quad x \geq B, \label{constraint}
\end{align}
if such a level exists.  This \lev model was first solved by Hilberink and Rogers \cite{ Hilberink_Rogers_2002} for a special class of  \lev processes taking the form of an independent sum of a linear Brownian motion and a compound Poisson process with negative jumps (cf.\ (3.21) on page 245 of \cite{ Hilberink_Rogers_2002}).  Kyprianou and Surya \cite{Kyprianou_Surya_2007} later showed for a general spectrally negative process that the optimal bankruptcy level exists and is explicitly determined by applying continuous and smooth fit when $X$ is of bounded and unbounded variation, respectively.  Regarding the cases with both positive and negative jumps, Chen and Kou \cite{Chen_Kou_2009} and Dao and Jeanblanc \cite{Dao_Jeanblanc_2012}
solved for double exponential jump diffusion and Le Courtois and Quittard-Pinon \cite{Courtois_2008} solved for stable processes.

\subsection{Our generalization}

We now incorporate scale effects and generalize the model described above by allowing the loss fraction $\hat{\eta}$ and tax rebate rate $\hat{\gamma}$ dependent on $X$.   

First, we generalize the debt value \eqref{debt_constant_unit} to
\begin{align*}
d (x; B, t) := \E_x \Big[ \int_0^{t \wedge \tau_B^- } e^{-rs} \hat{\rho} \diff s \Big] + \E_x \Big[ e^{-rt} 1_{\{t < \tau_B^- \}}\Big] + \frac 1 P \E_x \Big[ e^{-r\tau_B^- + X_{\tau_B^-} } \left(1-\overline{\eta} \big(X_{\tau_B^-} \big) \right) 1_{\{\tau_B^- < t \}}\Big]
\end{align*}
where $\overline{\eta} (\cdot) \geq 0$ is the rate of loss at bankruptcy.  The total debt value becomes
\begin{align*}
\mathcal{D} (x; B) := \E_x\Big[ \int_0^{\tau_B^-} e^{-(r+m)t} \left( P \hat{\rho}+ p \right) \diff t\Big] +  \E_x \Big[ e^{-(r+m) \tau_B^- + X_{\tau_B^-} } \left(1-\overline{\eta} \big(X_{\tau_B^-} \big) \right) 1_{\{\tau_B^- < \infty \}} \Big].
\end{align*}
By setting  
\begin{align*}
\eta(y) := e^y \overline{\eta}(y), \quad y \in \R,
\end{align*}
we can write
\begin{multline}
\mathcal{D} (x; B) = \E_x\Big[ \int_0^{\tau_B^-} e^{-(r+m)t} \left( P \hat{\rho}+ p \right)\diff t\Big] \\ +  \E_x \Big[ e^{-(r+m) \tau_B^- + X_{\tau_B^-}} 1_{\{\tau_B^- < \infty \}} \Big] - \E_x \left[ e^{-(r+m) \tau_B^-}\eta \big(X_{\tau_B^-}\big) 1_{\{\tau_B^- < \infty \}} \right]. \label{debt_total}
\end{multline}
Here notice that $\eta(\cdot)$ denotes the total loss amount whereas $\overline{\eta}(\cdot)$ is the rate of loss relative to the asset value.  We allow $\overline{\eta}$ to be larger than $1$, which lets one to model, for example, the case $\eta$ is a constant;  see Section \ref{subsection_other_example} below.

We next generalize the tax rebates; the firm value with bankruptcy level $B$ is
\begin{align}
\mathcal{V}(x;B) &:= e^x + \E_x \Big[ \int_0^{\tau_B^-} e^{-rt} f(X_t) \diff t\Big] -  \E_x \left[ e^{-r \tau^-_B}  \eta \big(   X_{\tau_B^-} \big)  1_{\{\tau_B^- < \infty \}} \right], \label{value_total}
\end{align}
where $f(\cdot) \in [0, P \hat{\rho}]$ is the rate of tax rebates.  As we discuss in Remark \ref{remark_finiteness}(1) below, under Assumption \ref{assump_eta}, each expectation in \eqref{debt_total}-\eqref{value_total} is finite for all $x > B$ and hence the equity value, $\mathcal{E}(x;B) = \mathcal{V}(x;B) - \mathcal{D}(x;B)$, is well-defined.





\section{Solutions for the Spectrally negative \lev models} \label{section_solution}
In this section, we study the problem \eqref{equity}-\eqref{constraint} with the debt and firm values generalized to \eqref{debt_total}-\eqref{value_total} focusing on the case $X$ is a spectrally negative \lev process, or a \lev process with only negative jumps. We assume that $X$ is uniquely defined by  the \emph{Laplace exponent},
\begin{align}
\kappa(s) := \log \E_0 \left[ e^{s X_1} \right] =c s +\frac{1}{2}\sigma^2 s^2 +\int_{(
0,\infty)}(e^{-s x}-1+s x 1_{\{0 <x<1\}})\,\Pi(\diff x), \quad  {s \geq 0},  \label{laplace_spectrally_negative}
\end{align}
where $c \in\R$, $\sigma \geq 0$, and $\Pi$ is  a measure  on $(0,\infty)$ such that $\int_{(0,\infty)} (1  \wedge x^2) \Pi( \diff x)<\infty$.
We ignore the case $X$ is the negative of a subordinator (monotonically decreasing a.s.). 

The process $X$ has paths of bounded variation if and only if $\sigma = 0$ and $\int_{(0,\infty)} (1 \wedge x) \Pi(\diff x) < \infty$.  For more details about the spectrally negative \lev process, we refer the reader to, e.g.,\ \cite{Bertoin_1996, Kyprianou_2006}.


\subsection{Scale functions} \label{subsec:scale_functions}For our derivation of optimal solutions, we first rewrite \eqref{debt_total}-\eqref{value_total} using the scale function.
For a given spectrally negative L\'evy process with
Laplace exponent $\kappa$, there exists an increasing
function 
%
\begin{align*}
W^{(q)}: \R \rightarrow [0,\infty), \quad q\ge 0,
\end{align*}
such that $W^{(q)}(x)=0$ for all $x<0$ and
\begin{align*}
\int_0^\infty e^{-s x} W^{(q)}(x) \diff x = \frac 1
{\kappa(s)-q}, \qquad s > \Phi(q)
\end{align*}
where
\begin{align*}
\Phi(q) := \sup \left\{ s > 0: \kappa(s) = q \right\}, \quad q \geq 0.
\end{align*}
Here $\Phi(q)$ is strictly increasing in $q$ because $\kappa$ is  zero at the origin and strictly convex on $[0,\infty)$.

If $\tau_a^+$ is the first time the process goes above $a
> x
> 0$ and $\tau_0^-$ is the first time it goes below zero as a special
case of \eqref{default_time}, then we have
\begin{align*}
\E_x \Big[ e^{-q \tau_a^+} 1_{\left\{ \tau_a^+ < \tau_0^-, \, \tau_a^+ < \infty
\right\}}\Big] = \frac {W^{(q)}(x)}  {W^{(q)}(a)}  \quad \textrm{and} \quad \E_x \Big[ e^{-q  \tau_0^-} 1_{\left\{ \tau_a^+ >
 \tau_0^-, \, \tau_0^- < \infty \right\}}\Big] &= Z^{(q)}(x) - Z^{(q)}(a) \frac {W^{(q)}(x)}
{W^{(q)}(a)},
\end{align*}
where
\begin{align*}
Z^{(q)}(x) := 1 + q \overline{W}^{(q)}(x) \quad \textrm{with} \quad \overline{W}^{(q)}(x) := \int_0^x W^{(q)}(y) \diff y, \quad x \in \R.
\end{align*}

 In particular, $W^{(q)}$ is continuously differentiable on $(0,\infty)$ if $\Pi$ does not have atoms (see \cite{Lambert_2000}), and it is twice-differentiable on $(0,\infty)$ if $\sigma > 0$ (see \cite{Chan_2009}).  For the rest of this paper, we assume the former.
\begin{assump}
We assume that $\Pi$ does not have atoms.
\end{assump}

Fix $q > 0$.  The scale function increases exponentially; 
\begin{align}
W^{(q)} (x) \sim \frac {e^{\Phi(q) x}} {\kappa'(\Phi(q))} \quad
\textrm{as } \; x \rightarrow \infty.
\label{scale_function_asymptotic}
\end{align}

As in Lemmas 4.3-4.4 of  \cite{Kyprianou_Surya_2007}, for all $q > 0$,
\begin{align}
\begin{split}
W^{(q)} (0) &= \left\{ \begin{array}{ll} 0, & \textrm{if $X$ is of unbounded
variation} \\ \frac 1 {\mu}, & \textrm{if $X$ is of bounded variation}
\end{array} \right\}, \\ W^{(q)'} (0+) &=
\left\{ \begin{array}{ll}  \frac 2 {\sigma^2}, &\textrm{if } \sigma > 0 \\
\infty, & \textrm{if } \sigma = 0 \; \textrm{and} \; \Pi(0,\infty) = \infty \\
\frac {q + \Pi(0,\infty)} {\mu^2}, & \textrm{otherwise}
\end{array} \right\}, \end{split} \label{at_zero}
\end{align}
where $\mu := c + \int_{(0,1)} x \Pi(\diff x)$ that is finite for the case $X$ is of bounded variation.

\subsection{In terms of the scale function}
We  now rewrite \eqref{debt_total}-\eqref{value_total} using the scale function.
Toward this end, we introduce the following shorthand notations:
\begin{align*}
\Lambda^{(q)}(x; B) &:= \E_x  \big[ e^{-q \tau_B^-} \eta (X_{\tau_B^-} )1_{\{\tau_B^- < \infty \}} \big], \\
 \mathcal{M}^{(q)}(x;B) &:= \E_x \Big[ \int_0^{\tau_B^-} e^{-qt} f (X_t) \diff t\Big] \quad \textrm{and} \quad
 \mathcal{N}^{(q)}(x;B) := \E_x \Big[ \int_0^{\tau_B^-} e^{-qt}  \left( P \hat{\rho}+ p \right) \diff t\Big],
\end{align*}
for any $q > 0$ and $x > B$.  These functions admit semi-explicit expressions in terms of the scale function.  By Lemmas 2.1-2.3 of Egami and Yamazaki \cite{Egami-Yamazaki-2011}, for all $x > B$ and $q > 0$, we can write
\begin{align}
\begin{split}
\Lambda^{(q)}(x; B) &=   \eta(B) \Big[ Z^{(q)}(x-B) - \frac q {\Phi(q)} W^{(q)}(x-B) \Big]  - W^{(q)}(x-B) H^{(q)}(B)\\ 
&\qquad + \int_{(0,\infty)} \Pi(\diff u) \int_0^{u \wedge (x-B)} W^{(q)} (x-z-B) [\eta(B)-\eta(z+B-u)] \diff z,\\
\mathcal{M}^{(q)}(x;B)  &=
W^{(q)} (x-B) G^{(q)}(B) -
\int_B^x W^{(q)} (x-y) f(y) \diff y, \\
\mathcal{N}^{(q)}(x;B) &= \left( P \hat{\rho}+ p \right) \Big( \frac 1 {\Phi(q)} W^{(q)}(x-B) - \overline{W}^{(q)}(x-B) \Big),
\end{split} \label{deft_Lambda_M}
\end{align}
where
\begin{align*}
H^{(q)}(B) &:= \int_{(0,\infty)}\Pi (\diff u)  \int_0^{u} e^{-\Phi(q) z} [\eta(B) - \eta(B-u+z)] \diff z, \\
G^{(q)}(B) &:= \int_0^\infty e^{-\Phi(q) y} f(y+B)\diff y.
\end{align*}

On the other hand, by Lemma 4.7 of \cite{Kyprianou_Surya_2007}, we have $\E_y \big[ e^{-q \tau_0^- + X_{\tau_0^-}} 1_{\{\tau_0^- < \infty \}} \big] = e^y -  \Gamma^{(q)}(y)$ 
where
\begin{align}
 \Gamma^{(q)}(y) := \frac {\kappa(1)-q} {1 - \Phi(q)} W^{(q)}(y) + (\kappa(1)-q) e^y \int_0^y e^{-z}W^{(q)} (z) \diff z, \quad q > 0 \; \textrm{and} \; y > 0. \label{def_Gamma}
\end{align}
Hence, for any $x > B$,
\begin{align*}
\E_x \Big[ e^{-(r+m) \tau_B^- + X_{\tau_B^-}} 1_{\{\tau_B^- < \infty \}} \Big] = e^B \E_{x-B} \Big[ e^{-(r+m) \tau_0^- + X_{\tau_0^-}} 1_{\{\tau_0^- < \infty \}} \Big] = e^x - e^{B}\Gamma^{(r+m)}(x-B).
\end{align*}

Putting altogether, for all $x > B$, we simplify \eqref{debt_total}-\eqref{value_total} to
\begin{align*}
\mathcal{D} (x; B) &=    e^x - e^{B}\Gamma^{(r+m)}(x-B) + \mathcal{N}^{(r+m)}(x;B)- \Lambda^{(r+m)}(x;B) , \\
\mathcal{V}(x;B) &= e^x  +  \mathcal{M}^{(r)}(x;B)- \Lambda^{(r)}(x;B),
\end{align*}
and we obtain the equity value
\begin{align}
\mathcal{E} (x; B) 
&= e^{B}\Gamma^{(r+m)}(x-B) + ( \mathcal{M}^{(r)}(x;B)  - \Lambda^{(r)}(x;B) ) - ( \mathcal{N}^{(r+m)}(x;B)   - \Lambda^{(r+m)}(x;B) ). \label{equity_in_terms_of_scale}
\end{align}

%

In view of \eqref{deft_Lambda_M}, because $f$ is bounded,  $\mathcal{M}^{(r)}(x;B)$ is finite for all $x > B$.
Regarding $\eta$, we assume the following for the rest of this section.
\begin{assump} \label{assump_eta}
We assume that $\eta$ is $C^2 (\R)$ and is bounded on $(-\infty,M]$ for any fixed $M \in \R$.  
\end{assump}
Here the $C^2$-assumption is imposed for simplicity of the arguments; this can be relaxed as discussed in Remark \ref{remark_twice_diff} below.  

\begin{remark} \label{remark_finiteness}
\begin{enumerate} \item  By Assumption \ref{assump_eta}, the equity value  $\mathcal{E}(x;B)$ is well-defined for any $x > B$.  \item
Because $\eta$ is continuous, $\Lambda^{(q)}(x; B)$, $\mathcal{N}^{(q)}(x;B)$ and $\mathcal{M}^{(q)}(x;B)$ are continuous in $B$ on $(-\infty, x]$ for any fixed $x \in \R$.
\end{enumerate}
\end{remark}

\subsection{Derivative with respect to $B$}
To derive the candidate bankruptcy level, we use the results in Egami and Yamazaki \cite{Egami-Yamazaki-2011} and obtain the derivative of $\mathcal{E} (x; B)$ with respect to $B$. Define
\begin{align}
\Theta^{(q)}(x) := W^{(q)'}(x) - \Phi(q) W^{(q)}(x), \quad x > 0 \; \textrm{and} \; q > 0, \label{def_Theta}
\end{align}
which is known to be always positive.   Given the descending ladder process $\{(\widehat{L}_t^{-1}, \widehat{H}_t): t \geq 0\}$  (see Section 6.2 of \cite{Kyprianou_2006}) and the dirac measure $\delta_0$ at $0$, we have $\int_0^\infty \E_0 \big[ e^{-q \widehat{L}_t^{-1}} 1_{\{\widehat{H}_t \in \diff x\}}\big] \diff t = \Theta^{(q)}(x) \diff x + W^{(q)} (0) \delta_0 (\diff x)$,
which can be confirmed by taking its Laplace transform and the Wiener-Hopf factorization.
In view of this, $\Theta^{(q)}(x)$ is monotonically decreasing in $q$ for every fixed $x > 0$.

The derivatives of $\Lambda^{(q)}(x;B)$, $\mathcal{M}^{(q)}(x;B)$ and  $\mathcal{N}^{(q)}(x;B)$ with respect to $B$ require technical details.  However, as shown by \cite{Egami-Yamazaki-2011}, each term can be expressed as a product of $\Theta^{(q)}(x-B)$ and some function of $B$ (that is independent of $x$).  For the proof of the following, see the proof of Proposition 3.1 of \cite{Egami-Yamazaki-2011}.
\begin{lemma}
 \label{lemma_M_derivative_B}
For every $x > B$  and $q > 0$, we have
\begin{align*}
\frac \partial {\partial B} \Lambda^{(q)}(x; B) &= \Theta^{(q)}(x-B)  \Big[ \frac q {\Phi(q)} \eta(B) + H^{(q)}(B) + \frac {\sigma^2} 2 \eta'(B) \Big], \\
\frac \partial {\partial B} \mathcal{M}^{(q)}(x; B) &= -\Theta^{(q)}(x-B)  G^{(q)}(B),  \quad \textrm{and} \quad
\frac \partial {\partial B} \mathcal{N}^{(q)}(x; B) = -\Theta^{(q)}(x-B)  \frac {P \hat{\rho} + p} {\Phi(r+m)}.
\end{align*}
\end{lemma}

For the derivative of \eqref{equity_in_terms_of_scale} with respect to $B$, we further obtain the following.
\begin{lemma}  \label{lemma_Gamma_derivative_B}
For every $x > B$, $ \partial (e^B \Gamma^{(r+m)}(x-B)) / {\partial B}
= -\frac {\kappa(1)-(r+m)} {1 - \Phi(r+m)} e^B \Theta^{(r+m)}(x-B)$.
\end{lemma}

By combining the two lemmas above, we obtain the derivative of $\mathcal{E}(x;B)$ with respect to $B$.  For all $B \in \R$, define
\begin{align}
J^{(r,m)}(B)  &:= \Big( \frac {r+m} {\Phi(r+m)} - \frac r {\Phi(r)} \Big)\eta(B)  - \left( H^{(r)}(B) - H^{(r+m)}(B) \right) \label{def_J}
\end{align}
and
\begin{align}
\begin{split}
K_1^{(r,m)}(B) &:=  
\frac {\kappa(1)-(r+m)} {1 - \Phi(r+m)} e^B - \frac {P \hat{\rho} + p} {\Phi(r+m)} + G^{(r)}(B) - J^{(r,m)}(B), \\
K_2^{(r)}(B) &:= G^{(r)}(B) + \frac r {\Phi(r)} \eta(B) +  H^{(r)}(B) + \frac {\sigma^2} 2 \eta'(B).
\end{split} \label{function_K}
\end{align}

\begin{proposition} \label{prop_derivative_B}
For every $x > B$,
\begin{align}
\frac {\partial} {\partial B}\mathcal{E} (x; B) = - \big[ \Theta^{(r+m)}(x-B) K_1^{(r,m)}(B) + \{\Theta^{(r)}(x-B) - \Theta^{(r+m)}(x-B) \} K_2^{(r)}(B) \big]. \label{derivative_B}
\end{align}

\end{proposition}

\begin{remark} \label{remark_K_2}If $\eta(B)$ is increasing in $B$, then $K_2^{(r)}$ is uniformly positive. 
\end{remark}

The following remark will be especially useful in our analysis; the proof is given in the Appendix.
\begin{remark} \label{remark_j}
We can also write
\begin{align*}
J^{(r,m)}(B) =  \frac 1 2 \sigma^2( \Phi(r+m) - \Phi(r)) \eta(B) + \int_{(0,\infty)} \Pi (\diff u)   \int_0^{u} (e^{-\Phi(r) z}-e^{-\Phi(r+m) z}) \eta(B-u+z)\diff z.
\end{align*}
Consequently, because $\Phi(q)$ is increasing in $q$ and $\eta(\cdot)$ is positive by assumption, $J^{(r+m)}(B) \geq 0$ for all $B \in \R$.
\end{remark}

\subsection{Continuous fit}  Before discussing the optimality, we consider the continuous fit condition:  
\begin{align*}
\mathcal{E}(B+;B) = 0.
\end{align*}
By taking $x \downarrow 0$ in \eqref{equity_in_terms_of_scale},
\begin{align}
\mathcal{E} (B+; B) 
&= e^{B}\Gamma^{(r+m)}(0) + ( \mathcal{M}^{(r)}(B+;B)  - \Lambda^{(r)}(B+;B) ) - ( \mathcal{N}^{(r+m)}(B+;B)   - \Lambda^{(r+m)}(B+;B) ).
\label{E_zero}
\end{align}
Here we have, by \eqref{def_Gamma}, $\Gamma^{(r+m)}(0) = \frac {\kappa(1)-(r+m)} {1 - \Phi(r+m)} W^{(r+m)}(0)$,
and by Proposition 3.2 of \cite{Egami-Yamazaki-2011}, for both $q=r$ and $q=r+m$,
\begin{align}
\begin{split}
\Lambda^{(q)}(B+; B) &= - W^{(q)} (0)  \Big( \frac q {\Phi(q)} \eta(B) + H^{(q)} (B) \Big) + \eta(B), \\
\mathcal{M}^{(q)} (B+; B) &= W^{(q)} (0) G^{(q)}(B)  \quad \textrm{and} \quad
\mathcal{N}^{(q)} (B+; B) = W^{(q)} (0) \frac {P \hat{\rho} + p} {\Phi(r+m)}.
\end{split} \label{Lambda_M_zero}
\end{align}

Substituting  \eqref{Lambda_M_zero} in \eqref{E_zero} and because $W^{(r)}(0) = W^{(r+m)}(0)$ as in \eqref{at_zero}, we obtain 
\begin{align}
\mathcal{E} (B+; B)  = W^{(r+m)} (0) K_1^{(r,m)}(B). \label{cont_fit_equation}
\end{align}
By \eqref{at_zero}, we conclude that for the bounded variation case the continuous fit condition is equivalent to $K_1^{(r,m)}(B) = 0$, while for the unbounded variation case it always holds no matter how $B$ is chosen.

\begin{remark}  \label{remark_smooth_fit} One can further pursue smooth fit for the case $X$ is of unbounded variation.  However, we do not discuss it here because it is not necessary for the proof of optimality in the subsequent sections.  In fact, it is expected that the smooth fit condition $\mathcal{E}'(B+;B)=0$ is equivalent to $K_1^{(r,m)}(B)=0$, and the optimal solution is expected to satisfy smooth fit (at least when $\sigma > 0$ by the results obtained in \cite{Egami-Yamazaki-2011}).  Our numerical results in Section  \ref{section_numerics} verifies that this is indeed so.
\end{remark}
\subsection{Optimality}
We assume that  there exists $B^* \in [-\infty, \infty]$ such that
\begin{align}
K_1^{(r,m)}(B) &\geq 0 \Longleftrightarrow B \geq B^*,  \label{cond_K_1}\\
K_2^{(r)}(B) &\geq 0, \quad B \geq B^*,\label{cond_K_2}
\end{align}
and prove its optimality.   Note that, if $B^* \in (-\infty, \infty)$, then we must have $K_1^{(r,m)}(B^*)  = 0$.
When the function $K_1^{(r,m)}$ stays constant on some interval, then the solution to $K_1^{(r,m)}(B) = 0$ may not be unique; in such cases, in view of \eqref{cond_K_1}, it is understood that the smallest element is chosen for $B^*$. 
We later discuss sufficient conditions that guarantee \eqref{cond_K_1}-\eqref{cond_K_2}.  Notice here that \eqref{cond_K_2} always holds given that $\eta(B)$ is increasing  in view of Remark \ref{remark_K_2}.  

We first show via continuous fit and \eqref{cond_K_1}-\eqref{cond_K_2} that any feasible bankruptcy level must be at least as large as $B^*$.  Toward this end, we use the following lemma.

\begin{lemma} \label{lemma_Theta_convergence}
When $X$ is of unbounded variation, we have $\Theta^{(r)}(y) - \Theta^{(r+m)}(y) \rightarrow 0$ as $y \downarrow 0$.  
\end{lemma}

\begin{lemma} \label{lemma_above_B_star}
If $B^* \in [-\infty, \infty]$ satisfies \eqref{cond_K_1} and $B$ satisfies \eqref{constraint}, then $B \in [B^*, \infty)$.
\end{lemma}

If $B^* = \infty$, then it is the only choice that satisfies the limited liability constraint and hence $B^*=\infty$ is trivially optimal.  
Suppose $B^* \in [-\infty, \infty)$. By how $B^*$ is chosen, \eqref{cond_K_1}-\eqref{cond_K_2} and the positivity of both $\Theta^{(r)}(y) - \Theta^{(r+m)}(y)$ and  $\Theta^{(r)}(y)$ for any $y > 0$,  Proposition \ref{prop_derivative_B} implies
\begin{align}
{\partial} \mathcal{E} (x; B) / {\partial B}\leq 0, \quad B^* < B < x. \label{equity_derivative}
\end{align}
Moreover, $B^*$ satisfies the limited liability constraint \eqref{constraint}.  Indeed, when $B^* =-\infty$, it is trivially satisfied; when $B^* > -\infty$, for any arbitrary $x >B^*$, we have by  \eqref{cont_fit_equation} that $0 \leq W^{(r+m)} (0) K_1^{(r,m)}(x) = \mathcal{E} (x+; x)  \leq \mathcal{E}(x;B^*)$.  
Here the first inequality holds because $K_1^{(r,m)}(x) \geq 0$ for any $x>B^*$ by \eqref{cond_K_1} and in particular it holds by equality for the unbounded variation case; the last inequality holds by  \eqref{equity_derivative}.
This together with the lemma above shows the optimality of $B^*$.  In summary, we have the following.

\begin{theorem} \label{theorem_main}
If there exists $B^* \in [-\infty, \infty]$ such that  \eqref{cond_K_1}-\eqref{cond_K_2} hold, then $B^*$ is the optimal bankruptcy level.
\end{theorem}

\begin{remark} \label{remark_twice_diff}
The assumption that $\eta$ is twice-differentiable on $\R$ as in Assumption \ref{remark_K_2} can be relaxed.  Its twice-differentiability at a fixed $B \in \R$  is required for Lemma \ref{lemma_M_derivative_B}.  In view of the arguments in this section (and in particular Remark \ref{remark_finiteness}-(2)), we only need it to be piecewise $C^2$ as long as it is continuous.
\end{remark}

\begin{remark}
The case $B^* = -\infty$ corresponds to the case it is never optimal to go bankrupt. The corresponding equity value becomes $\mathcal{E}(x; -\infty) := \lim_{B \downarrow -\infty}\mathcal{E}(x; B)$, which exists thanks to \eqref{equity_derivative}. In fact, by Assumption \ref{assump_eta}, the bankruptcy cost vanishes as $B \downarrow -\infty$ and because $\tau_B^- \xrightarrow{B \downarrow -\infty} \infty$ a.s., dominated and monotone convergence yields
\begin{align*}
\mathcal{E}(x; -\infty) = e^x + \E_x \left[ \int_0^\infty e^{-rt} f(X_t) \diff t\right] - \frac {P \hat{\rho} + p} {r+m},
\end{align*}
where the middle term can be written in terms of the scale function as in Corollary 8.9 of \cite{Kyprianou_2006}.
\end{remark}

\section{Sufficient Conditions} \label{section_example}

In the last section, we showed that the conditions \eqref{cond_K_1}-\eqref{cond_K_2} guarantee the optimality of the bankruptcy level $B^*$.  Here we obtain more concrete and economically sound conditions that satisfy \eqref{cond_K_1}-\eqref{cond_K_2}.  We first show that Assumption \ref{assump_optimality} below is sufficient and also encompasses the model by \cite{Hilberink_Rogers_2002, Kyprianou_Surya_2007} as reviewed in Section \ref{subsection_model_classic}.   We then give another example with a constant $\eta$ that does not satisfy Assumption \ref{assump_optimality} but nonetheless guarantees the optimality.  It is emphasized here that the assumptions discussed in this section are sufficient and clearly not necessary; \eqref{cond_K_1}-\eqref{cond_K_2} are expected to hold more generally.

\subsection{A sufficient condition}We show that the following assumption guarantees \eqref{cond_K_1}-\eqref{cond_K_2} and hence the optimality of $B^*$ holds by Theorem \ref{theorem_main}.
\begin{assump} \label{assump_optimality}
 (1) $\eta$ is increasing, (2) $\overline{\eta}$ is decreasing, (3) $f$ is increasing, and  (4) $0 \leq \overline{\eta}(\cdot) \leq 1$.
\end{assump}

Each condition in the assumption is justified by the empirical studies as we described in Section \ref{section_introduction}.  The monotonicity of $\eta$ and $\overline{\eta}$ in (1) and (2)  means that, as the firm's asset value increases, the amount of total bankruptcy costs increases and its proportion relative to the asset value decreases.   This is exactly the same as the results on scale effects derived by \cite{Ang_et_al_1982, Bris_2006, Warner_1976}; see Section \ref{section_introduction} and in particular Figure \ref{figure_bankruptcy}.  The monotonicity of $f$ (or the tax rebate rate) in (3) is implied by the convexity of the tax function that has been empirically confirmed by  \cite{Graham_Smith_1999}; the marginal effect of taxable income reduction (or the slope of the tax function) is an increasing function that converges to the full tax rate at infinity. The last condition (4) requires that the bankruptcy costs should not be more than the total asset value.


%
%

We first note that (1) guarantees \eqref{cond_K_2} by Remark \ref{remark_K_2}.  The following proposition shows that $K_1^{(r,m)}(B)$ is monotonically increasing and hence \eqref{cond_K_1} also holds.

\begin{proposition} \label{prop_based_on_assump}
Suppose Assumption \ref{assump_optimality} holds.  Then (i) there exists $B^* \in [-\infty, \infty)$ such that  \eqref{cond_K_1}-\eqref{cond_K_2} hold, and (ii)  it is an optimal bankruptcy level.
\end{proposition}


In order to prove Proposition \ref{prop_based_on_assump},  we first rewrite, in view of \eqref{function_K} and Remark \ref{remark_j},
\begin{align*}
K_1^{(r,m)}(B) &=  
e^B  l(B) - \frac {P \hat{\rho} + p} {\Phi(r+m)} + G^{(r)}(B)
\end{align*}
where 
\begin{multline}
l(B) := \frac {\kappa(1)-(r+m)} {1 - \Phi(r+m)} - \frac 1 2 \sigma^2 ( \Phi(r+m) - \Phi(r))\overline{\eta}(B)  \\  - \int_{(0,\infty)} \Pi (\diff u)   \int_0^{u} (e^{-\Phi(r) z}-e^{-\Phi(r+m) z}) e^{z-u}\overline{\eta}(B-u+z)\diff z.
\end{multline}
Because $\overline{\eta}(B)$ is decreasing in $B$ by assumption and $\Phi(q)$ is increasing in $q$,  we have that $l(B)$ is increasing.   Because $\overline{\eta}(B)$ is monotone and bounded in $[0,1]$, there exists $\overline{\eta}(-\infty) := \lim_{B \downarrow -\infty} \overline{\eta}(B) \in [0,1]$.  By the monotone convergence theorem, we obtain
\begin{multline*}
\lim_{B \downarrow -\infty}l(B) = \frac {\kappa(1)-(r+m)} {1 - \Phi(r+m)} - \frac 1 2 \sigma^2  (\Phi(r+m) - \Phi(r)) \overline{\eta}(-\infty)  \\   \qquad - \int_{(0,\infty)} \Pi (\diff u)  \int_0^{u} (e^{-\Phi(r) z}-e^{-\Phi(r+m) z}) e^{z-u}\overline{\eta}(-\infty)\diff z = \frac {\kappa(1)-(r+m)} {1 - \Phi(r+m)} - \overline{\eta}(-\infty)  j^{(r,m)}
\end{multline*}
where
\begin{align*}
j^{(r,m)} &:= \frac 1 2 \sigma^2 ( \Phi(r+m) - \Phi(r)) +  \int_{(0,\infty)} \Pi (\diff u)  e^{-u} \Big(  \frac {1 - e^{-(\Phi(r)-1)u}} {\Phi(r)-1} - \frac {1 - e^{-(\Phi(r+m)-1)u}} {\Phi(r+m)-1} \Big). 
\end{align*}
\begin{lemma} \label{lemma_small_j}
We have $j^{(r,m)}
= \frac {\kappa(1) - (r+m)} {1-\Phi(r+m) } - \frac {\kappa(1)-r} {1-\Phi(r)}$.
\end{lemma}

Now by Lemma  \ref{lemma_small_j}, $\lim_{B \downarrow -\infty}l(B) = (1-\overline{\eta}(-\infty)) \frac {\kappa(1)-(r+m)} {1-\Phi(r+m) } + \overline{\eta}(-\infty) \frac {\kappa(1)-r} {1-\Phi(r)}$.
Because $0 \leq \overline{\eta}(-\infty) \leq 1$ and $\kappa$ is convex on $[0,\infty)$ and zero at the origin,  we have $\lim_{B \downarrow -\infty} l(B) > 0$ and hence $l(B) > 0$ for any $B \in \R$. Consequently, $e^B l(B)$ is strictly increasing in $B$.  Finally, $- \frac {P \hat{\rho} + p} {\Phi(r+m)} + G^{(r)}(B)$ is increasing in $B$ because $f$ is increasing by assumption.  Therefore, there exists a unique $B^*$ that satisfies \eqref{cond_K_1}.  Here $K_1^{(r,m)}(B) \xrightarrow{B \uparrow \infty} \infty$, and hence $B^* < \infty$. As is discussed above, \eqref{cond_K_2} holds by the monotonicity of $\eta$.
Now, by Theorem \ref{theorem_main}, Proposition \ref{prop_based_on_assump} holds.

\begin{remark} \label{remark_vanishing_G}
If $G^{(r)}(B) \xrightarrow{B \downarrow -\infty} 0$ (which holds when the tax rebate function $f(x)$ vanishes as $x \downarrow -\infty$ by dominated convergence), then $K_1^{(r,m)}(B) \xrightarrow{B \downarrow -\infty} -  {(P \hat{\rho} + p)} /{\Phi(r+m)} < 0$ and hence $B^* > -\infty$.
\end{remark}

%

\subsection{Reduction to the case by \cite{Hilberink_Rogers_2002,Kyprianou_Surya_2007}}  

As an example that satisfies Assumption \ref{assump_optimality}, we revisit the simple case by \cite{Hilberink_Rogers_2002,Kyprianou_Surya_2007} as reviewed in Section \ref{subsection_model_classic}.  Namely, we set
\begin{align*}
\quad f(x) = 1_{\{x \geq \log v_T \}} P \hat{\gamma} \hat{\rho}, \quad \textrm{and} \quad \overline{\eta}(B) = \hat{\eta},
\end{align*}
and confirm that our results match those of \cite{Hilberink_Rogers_2002,Kyprianou_Surya_2007}.

First,  Assumption \ref{assump_optimality} is trivially satisfied and hence the optimal threshold level $B^*$ is uniquely given by $K_1^{(r,m)}(B^*) = 0$ (here $B^* > -\infty$ because $f(x) \downarrow 0$ as $x \downarrow -\infty$ by Remark \ref{remark_vanishing_G}).
%
In this case,
\begin{multline*}
G^{(r)}(B) = \int_0^\infty  e^{-\Phi(r) y} f(y+B)   \diff y = P \hat{\gamma} \hat{\rho}   \int_{(\log v_T - B) \vee 0}^\infty  e^{-\Phi(r) y}   \diff y = \frac {P \hat{\gamma} \hat{\rho}} {\Phi(r)} e^{-\Phi(r)((\log v_T - B) \vee 0)}\\ = \frac {P \hat{\gamma} \hat{\rho}} {\Phi(r)} \left( e^{\log v_T - B} \vee 1 \right)^{-\Phi(r)} = \frac {P \hat{\gamma} \hat{\rho}} {\Phi(r)} \left( (e^{-B} v_T) \vee 1 \right)^{-\Phi(r)} = \frac {P \hat{\gamma} \hat{\rho}} {\Phi(r)} \Big( \frac {e^{B}}  {v_T} \wedge 1 \Big)^{\Phi(r)}.
\end{multline*}
We also have, by  Remark \ref{remark_j} and Lemma  \ref{lemma_small_j},
\begin{align*}
J^{(r,m)}(B)  &=  \frac 1 2 \sigma^2 ( \Phi(r+m) - \Phi(r)) e^B \hat{\eta} + \int_{(0,\infty)} \Pi (\diff u)  \int_0^{u} (e^{-\Phi(r) z}-e^{-\Phi(r+m) z}) e^{B-u+z} \hat{\eta} \diff z \\
&=  j^{(r,m)} \hat{\eta} e^B = \Big( \frac {\kappa(1)-(r+m)} {1-\Phi(r+m)} - \frac {\kappa(1)-r} {1-\Phi(r)} \Big) \hat{\eta} e^B.
\end{align*}
Combining the above, 
\begin{align*}
K_1^{(r,m)}(B) &=  
\frac {\kappa(1)-(r+m)} {1 - \Phi(r+m)} e^B - \frac {P (\hat{\rho} + m)} {\Phi(r+m)} + \frac {P \hat{\gamma} \hat{\rho}} {\Phi(r)} \Big( \frac {e^{B}}  {v_T} \wedge 1 \Big)^{\Phi(r)} - \hat{\eta} e^B \Big(  \frac {\kappa(1)-(r+m)} {1-\Phi(r+m)} - \frac {\kappa(1)-r} {1-\Phi(r)} \Big) \\
&=  - \frac {P (\hat{\rho} + m)} {\Phi(r+m)} + \frac {P \hat{\gamma} \hat{\rho}} {\Phi(r)} \Big( \frac {e^{B}}  {v_T} \wedge 1 \Big)^{\Phi(r)} +  e^B \Big( (1-\hat{\eta}) \frac {\kappa(1) - (r+m)} {1-\Phi(r+m)} + \hat{\eta} \frac {\kappa(1)-r} {1-\Phi(r)} \Big).
\end{align*}
The unique value of $B^*$ that satisfies $K_1^{(r,m)}(B^*) = 0$ indeed matches the result of \cite{Hilberink_Rogers_2002,Kyprianou_Surya_2007}.  

\subsection{Other examples}   \label{subsection_other_example}
Although Assumption \ref{assump_optimality} is already a reasonable assumption, it is expected that \eqref{cond_K_1}-\eqref{cond_K_2} hold more generally.  As an example where Assumption \ref{assump_optimality} is violated but the optimality of $B^*$ holds, we consider the case the value of bankruptcy costs is a constant, i.e., $\eta \equiv \eta_0$; see \cite{Leland_1994}.  In this case, we have $\overline{\eta}(y) = \eta_0 e^{-y}$, which violates Assumption \ref{assump_optimality}-(4).  Nonetheless, the optimality  trivially holds upon the monotonicity of $G^{(r)}$.  Indeed, $H^{(r)} \equiv H^{(r+m)} \equiv 0$ and hence, by \eqref{def_J}, $J^{(r,m)}(B)  = \big( \frac {r+m} {\Phi(r+m)} - \frac r {\Phi(r)} \big) \eta_0$,
which is a constant.  Now in view of the definition of $K_1^{(r+m)}$ in \eqref{function_K}, it is clearly increasing in $B$ for example when Assumption \ref{assump_optimality}-(3) holds, and hence \eqref{cond_K_1} is valid.  Moreover, \eqref{cond_K_2} trivially holds by Remark \ref{remark_K_2}.

\section{Numerical Studies} \label{section_numerics}

In this section, we conduct a series of numerical experiments using the results obtained in the previous sections.  
Here we focus on analyzing the impacts of scale effects on the capital structure because this is the main theme of this paper.  For other numerical experiments such as the impacts of the \lev measure and term structure, we refer the reader to a comprehensive study conducted by 
 \cite{Chen_Kou_2009,Dao_Jeanblanc_2012}.  We first illustrate how to compute the optimal bankruptcy level, the associated equity/debt/firm values as well as the optimal face value as a solution to the two-stage problem  \cite{Chen_Kou_2009,Leland_1994, Leland_Toft_1996}. 
We then study computationally the impacts of the scale effects on these values.
We use an example where Assumption \ref{assump_optimality} holds and, for $X$, we follow \cite{Hilberink_Rogers_2002} and use a mixture of a Brownian motion and a compound Poisson process with i.i.d.\ exponential jumps.   


For the bankruptcy costs, let
\begin{align}
\overline{\eta}(x) = \eta_0 \left( 1 \wedge e^{-a(x-b)} \right), \quad x \in \R, \label{eta_bar_numerics}
\end{align}
for some $0 \leq a \leq 1$, $b \in \R$ and $0 \leq \eta_0 \leq 1$.  This is clearly decreasing in $x$ and bounded in $[0,1]$. Moreover, $\eta(x) = \eta_0 \left( e^x \wedge e^{(1-a)x+ab} \right) $ is increasing.   We call $a$ the \emph{degree of bankruptcy cost concavity}.   When $a=0$, the scale effect vanishes  and the ratio $\overline{\eta}$ becomes a constant.  However, as $a$ gets larger, the concavity of the function $\eta$ increases.  When $a=1$, $\eta(x) = \eta_0 e^{x \wedge b}$ and is constant uniformly on $[b,\infty)$.  Clearly, the bankruptcy cost is monotonically decreasing in $a$.

For the tax benefits, let
\begin{align}
f(x) = P \hat{\gamma} \hat{\rho} \left( e^{x-c} \wedge 1 \right) \label{f_2_numerics}
\end{align} 
for some $c \in \R$.  While it is rather an oversimplification, this efficiently models the effective tax function, empirically obtained by \cite{Graham_Smith_1999}, which is strictly convex for small taxable income but is closer to linear for large taxable income.  Indeed, \eqref{f_2_numerics} is increasing on $(-\infty, c)$ and is constant on $[c,\infty)$, making its antiderivative  a desired convex function.  We call $c$ \emph{the degree of tax convexity}.  The value of tax benefits decreases monotonically as $c$ increases.

Our choice of \eqref{eta_bar_numerics}-\eqref{f_2_numerics} satisfy all the conditions in Assumption \ref{assump_optimality} and hence Proposition \ref{prop_based_on_assump} holds.


Regarding $X$, we consider the exponential diffusion in the form $X_t = \mu t + \sigma B_t - J_t$, $t \geq 0$, where $B$ is a standard Brownian motion and $J$ is a pure-jump \lev process with its \lev measure
\begin{align}
\Pi(\diff u) = \lambda \beta e^{-\beta u} \diff u, \quad u > 0. \label{lev_measure_exp}
\end{align}
Its Laplace exponent \eqref{laplace_spectrally_negative} is given by $\kappa(s) = \mu s + \frac 1 2 \sigma^2 s^2  + \lambda
\big( \frac {\beta} {\beta + s}  -1\big)$, $s \in \R$.
The scale function of this process has an explicit expression written in terms of a sum of exponential functions.
For every $q > 0$, there are two negative real roots $-\xi_{1,q}$ and $-\xi_{2,q}$ to the equation $\kappa(s)=q$ and, as is discussed in \cite{Egami_Yamazaki_2010_2}, its scale function is given by
\begin{align}
\begin{split}
W^{(q)}(x) &= -\sum_{i =1,2}  \frac 1 {\kappa'(-\xi_{i,q})} \left[ e^{\Phi(q) x} - e^{-\xi_{i,q}x} \right], \quad x \geq 0.
\end{split} \label{scale_function_exp}
\end{align}

By straightforward but tedious algebra, we can obtain each functional in the equity value  \eqref{equity_in_terms_of_scale} as well as the function $K_1^{(r,m)}$ in \eqref{function_K} that determines the optimal bankruptcy level.    

\begin{figure}[htbp]
\begin{center}
\begin{minipage}{1.0\textwidth}
\centering
\begin{tabular}{cc}
\includegraphics[scale=0.5]{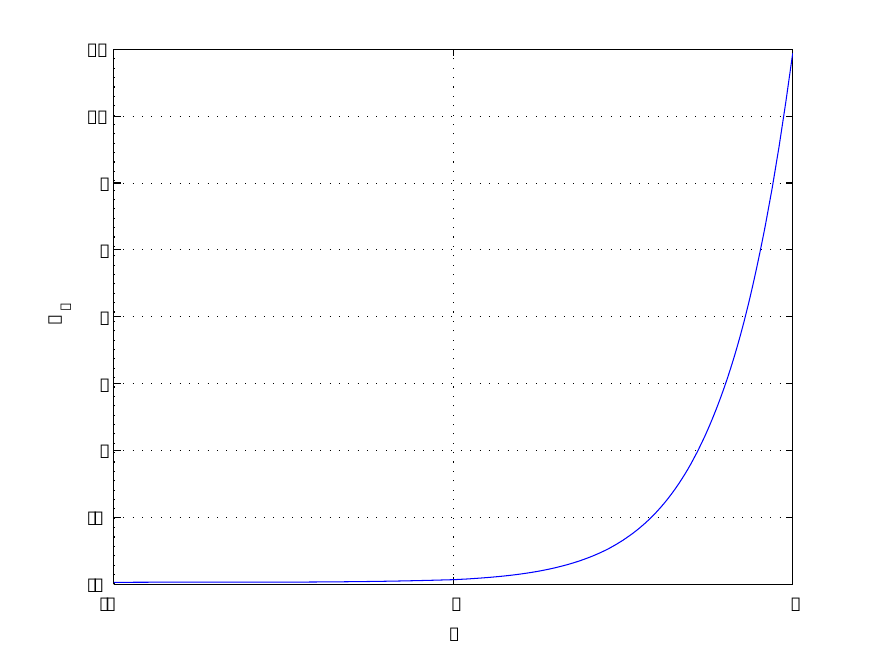}  & \includegraphics[scale=0.5]{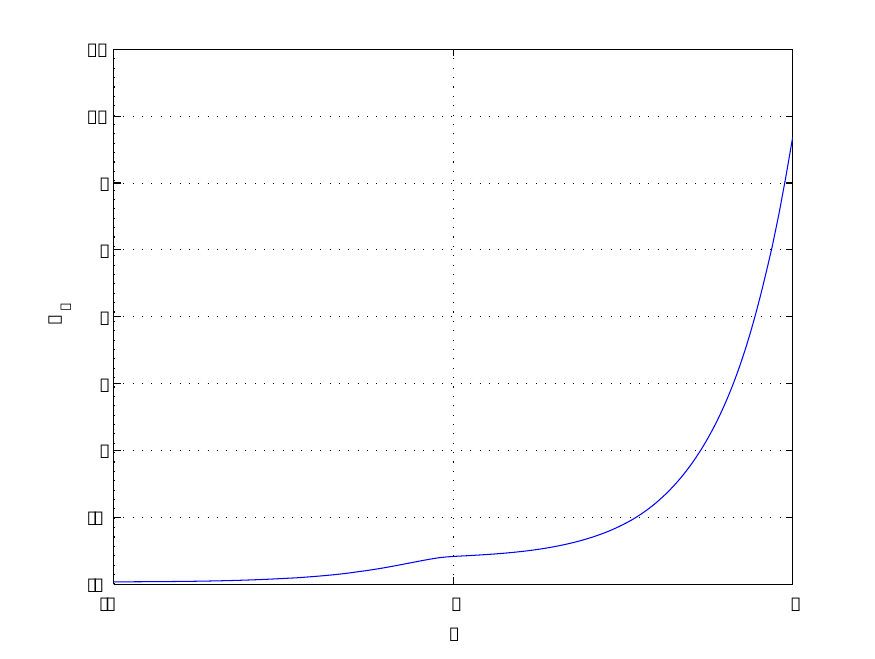} \\
case 1 & case 2
\end{tabular}
\end{minipage}
\caption{The plots of $K_1^{(r,m)}(B)$.   } \label{figure_K_1}
\end{center}
\end{figure}

We first illustrate how to compute the optimal bankruptcy level, firm/equity/debt values and the optimal capital structure.  We use $r=7.5\%$, $\delta=7\%$, $\hat{\gamma} = 35\%$, $\sigma = 0.2$, $\lambda=0.5$ and $\beta=9$ which were used in \cite{Hilberink_Rogers_2002, Kyprianou_Surya_2007, Leland_1994, Leland_Toft_1996}.   We also use $\hat{\rho} = 8.162\%$ and $m=0.2$, which were used in \cite{Chen_Kou_2009}.  We choose the drift term $\mu$ so that the martingale property $\kappa(1)=r-\delta$ is satisfied.  Regarding the parameters for $\eta$ and $f$ defined above, we consider the following two cases:
\begin{description}
\item[case 1]  $\eta_0=0.9$, $a = 0.5$, $b = 0$ and $c = 5$,
\item[case 2]  $\eta_0=0.5$, $a = 0.01$, $b = 5$ and $c = 0$.
\end{description}
In case 1, the slopes of $\overline{\eta}$ and $f$ are magnified by how the parameters are chosen.  On the other hand, in case 2, these values are constant at least when $x \in [0,5]$, making the model similar to \cite{Hilberink_Rogers_2002,Kyprianou_Surya_2007}.

\begin{figure}[htbp]
\begin{center}
\begin{minipage}{1.0\textwidth}
\centering
\begin{tabular}{cc}
\includegraphics[scale=0.5]{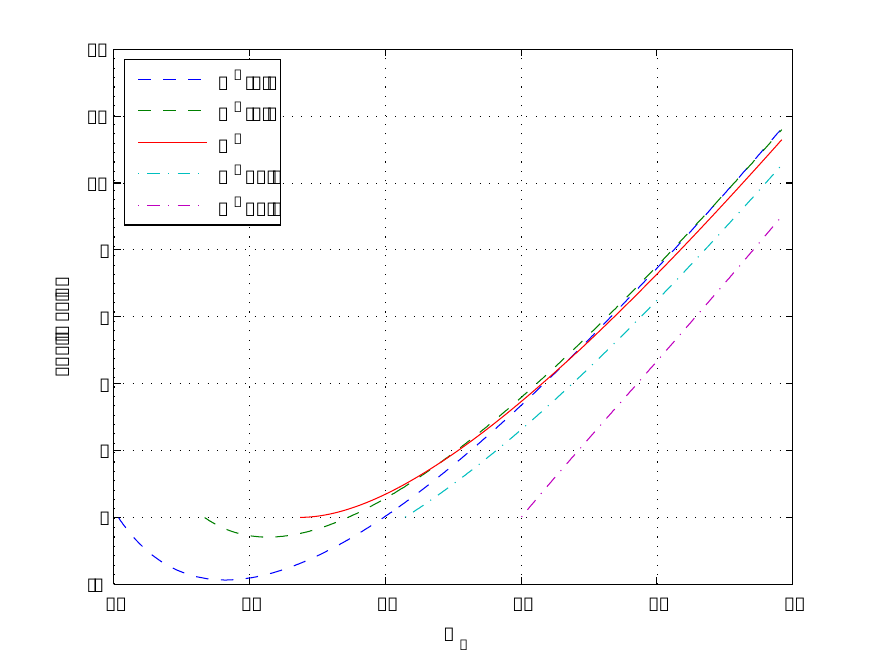}  & \includegraphics[scale=0.5]{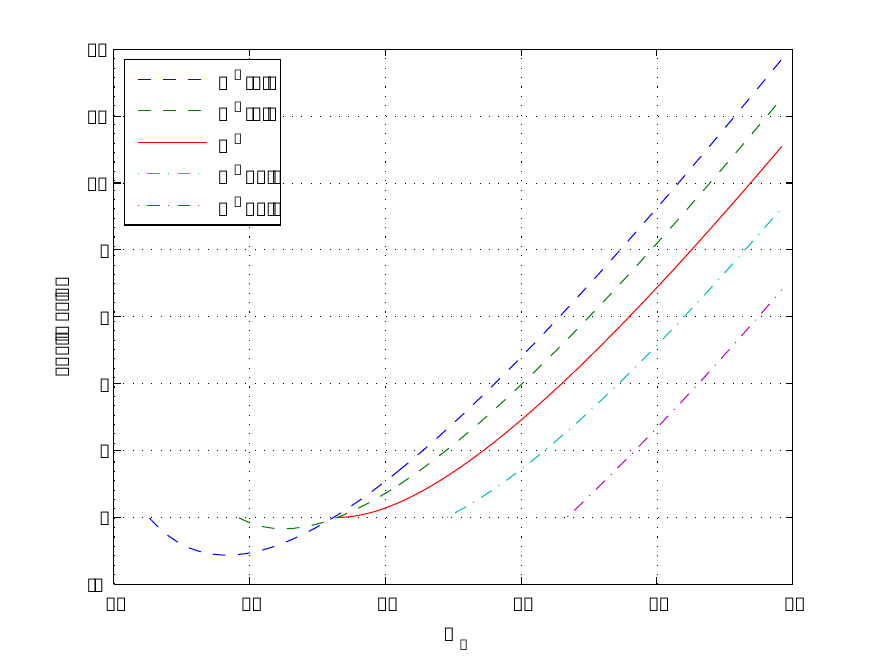} \\
case 1 & case 2
\end{tabular}
\end{minipage}
\caption{\small{The equity values as a function of $V_0$ for various values of $B$.}} \label{figure_value}
\end{center}
\end{figure}

\begin{figure}[htbp]
\begin{center}
\begin{minipage}{1.0\textwidth}
\centering
\begin{tabular}{cc}
\includegraphics[scale=0.5]{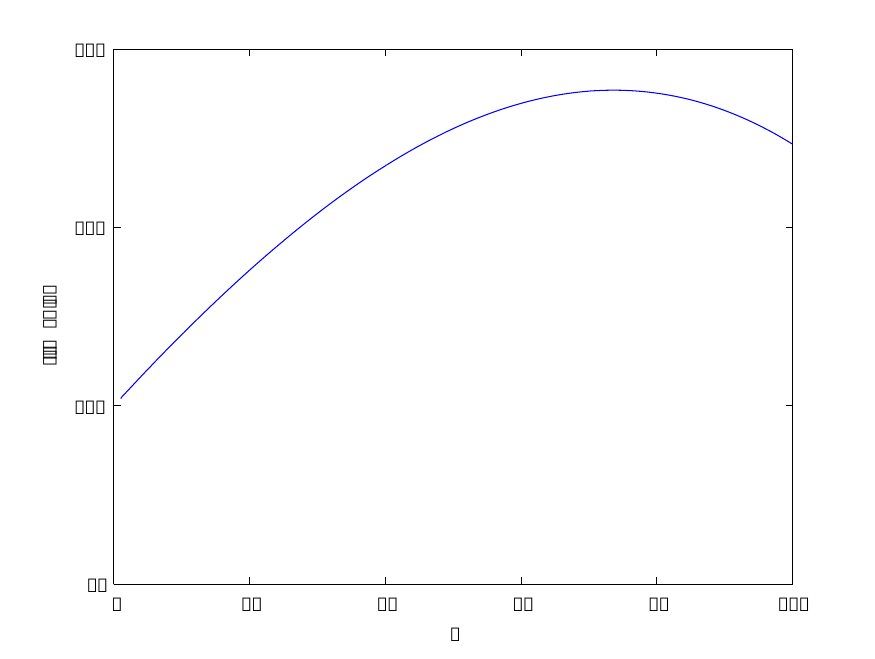}  & \includegraphics[scale=0.5]{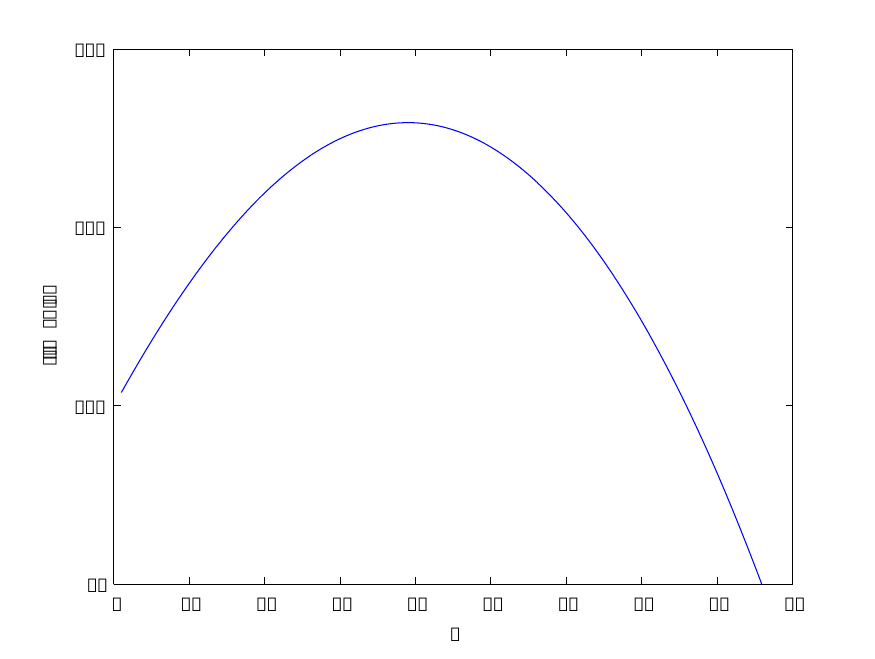} \\
case 1 & case 2
\end{tabular}
\end{minipage}
\caption{\small{The firm value as a function of $P$ for the two-stage problem. }} \label{figure_two_stage}
\end{center}
\end{figure}

Figure \ref{figure_K_1} shows the function $K_1^{(r+m)}$ in \eqref{function_K} as a function of $B$ when $P=50$.  As shown in the proof of Proposition \ref{prop_based_on_assump}, this is indeed monotonically increasing for both cases.   The optimal bankruptcy level $B^*$ can therefore be computed by the bisection method and we obtain $B^*=3.61$ and $B^* =3.64$ for cases 1 and 2, respectively.  In case 2, we notice a non-smooth point at zero and this is caused because $c=0$ in the definition of $f$ in \eqref{f_2_numerics}.  However, because $B^*$ is chosen larger than zero, the tax rate becomes constant until bankruptcy.   Furthermore, because $B^* < b$, the loss fraction $\overline{\eta}$ at bankruptcy ends up being a constant.  In case 1, on the other hand, the tax rate fluctuates over time and $\overline{\eta}$ is also asset-value-dependent.

Using the optimal bankruptcy levels $B^*$ computed above, we can compute the equity values.  In Figure \ref{figure_value}, we plot their values as a function of the asset value $V_0 = e^x \geq e^B$ for $B = B^*-0.2, B^*-0.1, B^*, B^*+0.1, B^*+0.2$.  As shown in Lemma \ref{lemma_above_B_star}, we can confirm that, when $B$ is taken lower than $B^*$,  it violates the limited liability constraint \eqref{constraint}.  For $B$ larger than $B^*$, the equity value is dominated by the value under $B^*$.  We note that continuous fit at $B$ always holds as in \eqref{cont_fit_equation}.  Also, as we have discussed in Remark \ref{remark_smooth_fit}, we also observe that smooth fit holds at $B^*$.

For the optimal capital structure, we solve the \emph{two-stage problem} as studied by \cite{Chen_Kou_2009,Leland_1994, Leland_Toft_1996} where the final goal is to choose $P$ (with $m$ and $x$ fixed constant) such that the firm's value $\mathcal{V}$ is maximized, namely, 
\begin{align}
\max_{P} \mathcal{V}(x; B^*(P), P) \label{two_stage_problem}
\end{align}
where we emphasize the dependency of $\mathcal{V}$ and $B^*$ on $P$.  As discussed in \cite{Chen_Kou_2009}, this is only a rough approximation of the optimal capital structure/bankruptcy strategy and does not fully capture the conflict of interest between debt holders and equity holders.  However, this is commonly used, and one major advantage of using this formulation is that \eqref{two_stage_problem} typically admits a global optimum $P^*$.   In particular, the concavity  of $\mathcal{V}(x; B^*(P), P)$ with respect to $P$ has been analytically shown by \cite{Chen_Kou_2009} for the double exponential jump diffusion case in their setting.  In our case with a generalized bankruptcy cost and tax rate functions, the proof of concavity is significantly difficult.  However, from our numerical results given below, it seems to hold at least with our specifications \eqref{eta_bar_numerics}-\eqref{f_2_numerics}; see Figure \ref{figure_two_stage} below.

We set $V_0 = 100$ (or $x = \log (100)$) and obtain $B^*$ for $P$ running from $0$ to $100$.  The corresponding firm value $\mathcal{V}$ is computed for each $P$ and $B^* = B^*(P)$, and is shown in Figure \ref{figure_two_stage}.  As can be seen, these are indeed concave and hence the optimal face values of debt $P^*=73.7$ and $P^*=39$ are obtained for cases 1 and 2, respectively.

\subsection{Numerical results on the scale effects}  We are now ready to study the impacts of scale effects on the optimal bankruptcy levels, equity/debt/firm values as well as the optimal capital structure.  We use the degree of bankruptcy cost concavity $a$ and tax convexity $c$ in \eqref{eta_bar_numerics}-\eqref{f_2_numerics} as proxies for scale effects  and study how these values change with $a$ and $c$.   The parameters for \eqref{eta_bar_numerics}-\eqref{f_2_numerics} are the same as case 1 above, unless specified otherwise.  We set $V_0 = 100$, and in particular, for Figures \ref{one_stage_bankruptcy} and \ref{one_stage_tax}, $P=50$.

\begin{figure}[htbp]
\begin{center}
\begin{minipage}{1.0\textwidth}
\centering
\begin{tabular}{cc}
\includegraphics[scale=0.4]{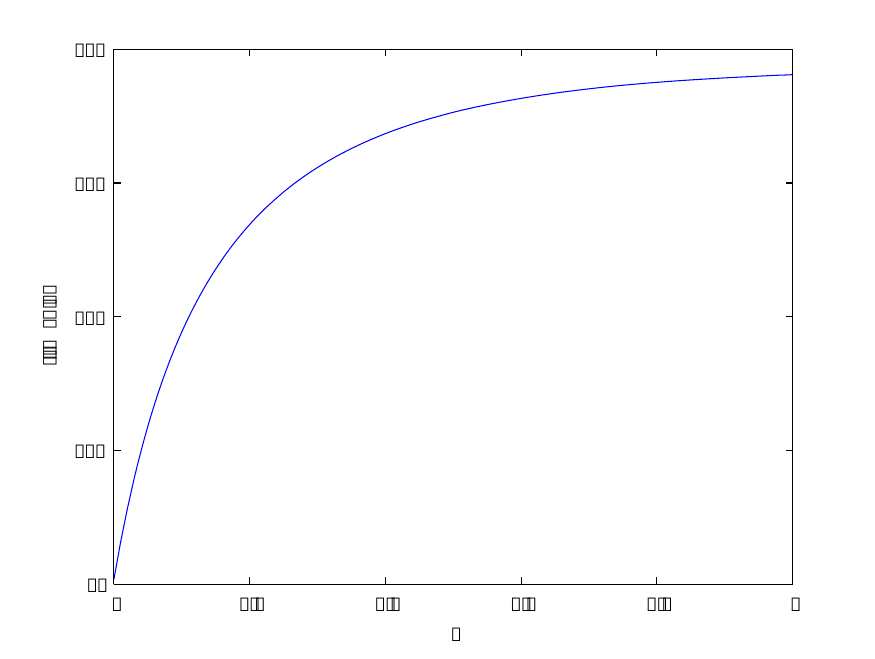} & \includegraphics[scale=0.4]{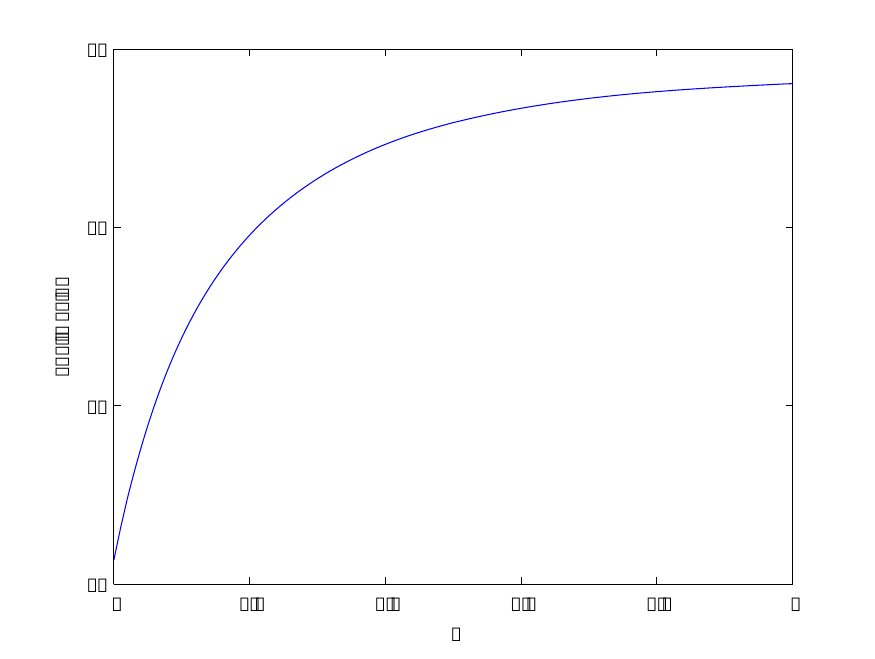}    \\
(i) firm value $\mathcal{V}(x; B^*)$ & (ii) equity value $\mathcal{E}(x; B^*)$   \\
 \includegraphics[scale=0.4]{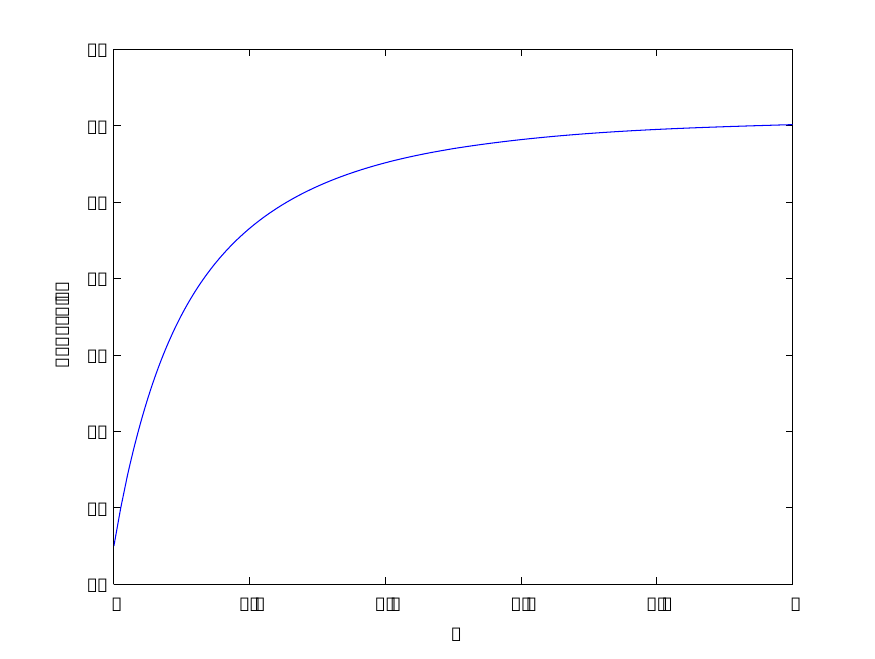} & \includegraphics[scale=0.4]{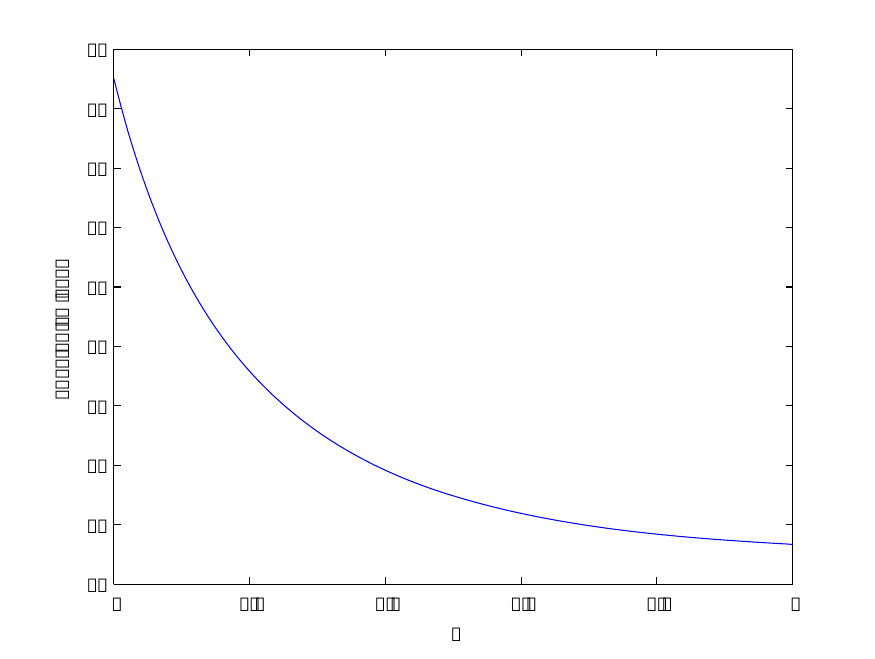}\\
 (iii) debt value $\mathcal{D}(x; B^*)$ & (iv) bankruptcy level $e^{B^*}$  
\end{tabular}
\end{minipage}
\caption{\small{Effects of bankruptcy cost concavity $a$ for fixed face value $P=50$.}}  \label{one_stage_bankruptcy}
\end{center}
\end{figure}


Figure \ref{one_stage_bankruptcy} shows the impacts of the bankruptcy cost concavity.  All the figures are with respect to $a$ running from $0$ to $1$.  Recalling that the bankruptcy cost is decreasing in $a$, most of the results are relatively straightforward.  As can be observed in (i)-(iii), each of the firm/equity/debt values is increasing.  From (iv), we see that  the bankruptcy level $B^*$ is decreasing.  This is because, as the bankruptcy cost decreases (or $a$ increases), the bankruptcy level can be lowered further without violating the limited liability constraint.  
Notice that all the figures here exhibit not only monotonicity but also concavity/convexity.

\begin{figure}[htbp]
\begin{center}
\begin{minipage}{1.0\textwidth}
\centering
\begin{tabular}{cc}
\includegraphics[scale=0.4]{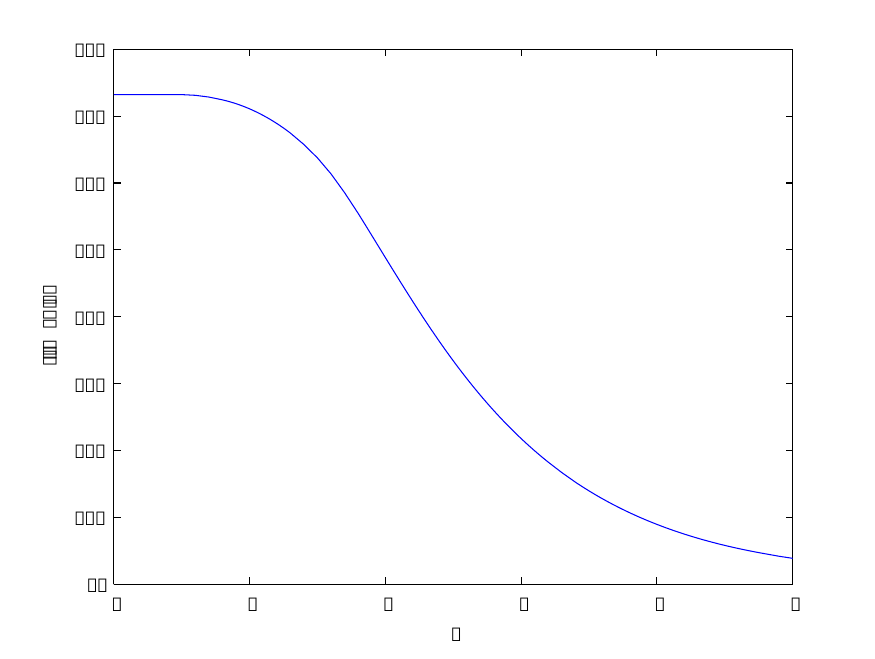} & \includegraphics[scale=0.4]{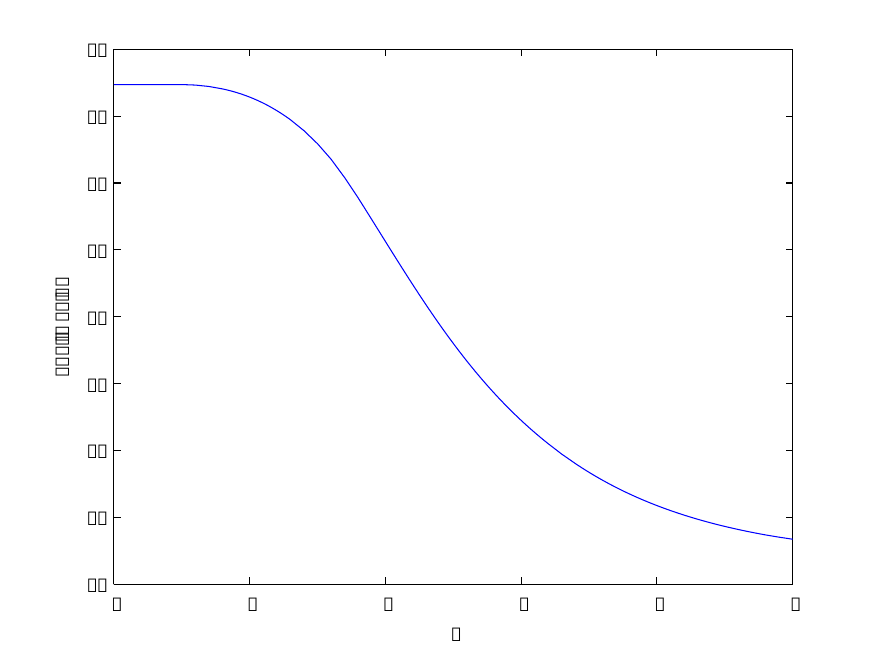}   \\
(i)  firm value $\mathcal{V}(x; B^*)$ & (ii) equity value $\mathcal{E}(x; B^*)$   \\
\includegraphics[scale=0.4]{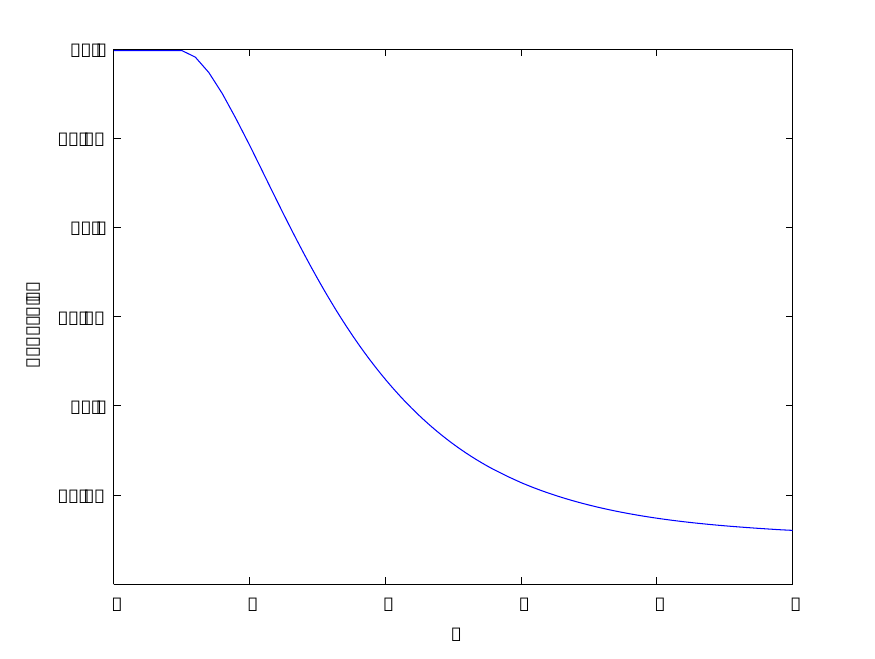} & \includegraphics[scale=0.4]{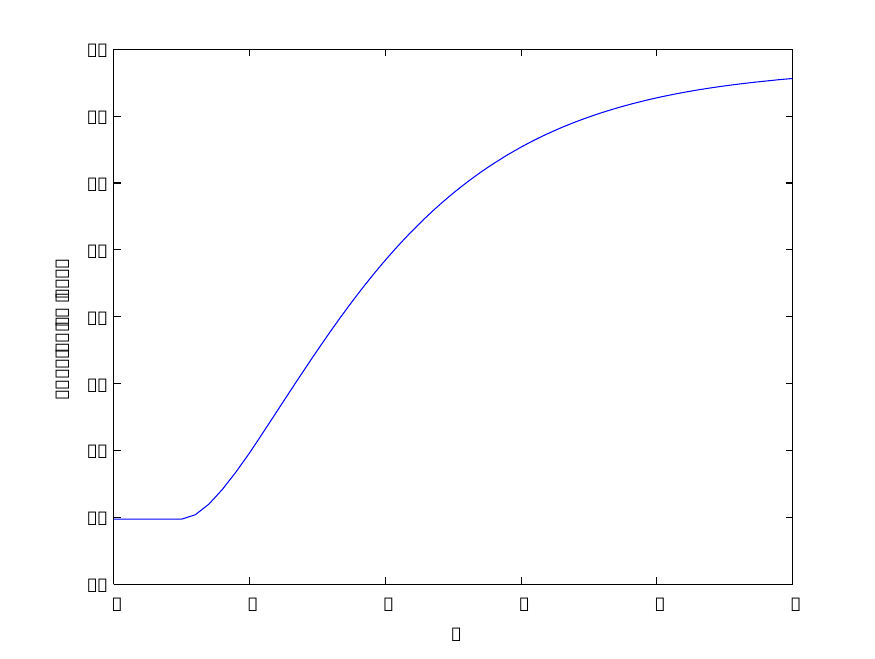} \\
 (iii) debt value $\mathcal{D}(x; B^*)$ & (iv) bankruptcy level $e^{B^*}$  
\end{tabular}
\end{minipage}
\caption{\small{Effects of tax convexity $c$ for fixed face value $P=50$.}}  \label{one_stage_tax}
\end{center}
\end{figure}

Figure \ref{one_stage_tax} shows the impacts of the tax convexity.  All the figures are with respect to $c$ running from $3$ to $8$.  Because it reduces the tax benefits monotonically, the results here are essentially opposite of those in Figure  \ref{one_stage_bankruptcy}.  Each of the firm/equity/debt values is decreasing.  The value of $B^*$ is increasing because it must be raised so as not to violate the limited liability constraint.  
Contrary to Figure  \ref{one_stage_bankruptcy}, these figures no longer exhibit concavity nor convexity.  This is due to the shape of the tax rate function \eqref{f_2_numerics}, which converges to zero as $c$ gets smaller and to the maximum tax benefit rate $P \hat{\gamma} \hat{\rho}$ as it gets larger.  This means that the marginal effect of changing the value of $c$ vanishes as $c$ gets sufficiently large or sufficiently small.   In particular, for $c$ smaller than the bankruptcy level $B^*$, because $f(x)$ is constant uniformly on $[c,\infty)$, changing the value of $c$ does not have any effect because $f(X_t)$ remains constant at any time before  bankruptcy.  This results in a flat line on each figure for small $c$.


\begin{figure}[htbp]
\begin{center}
\begin{minipage}{1.0\textwidth}
\centering
\begin{tabular}{cc}
\includegraphics[scale=0.4]{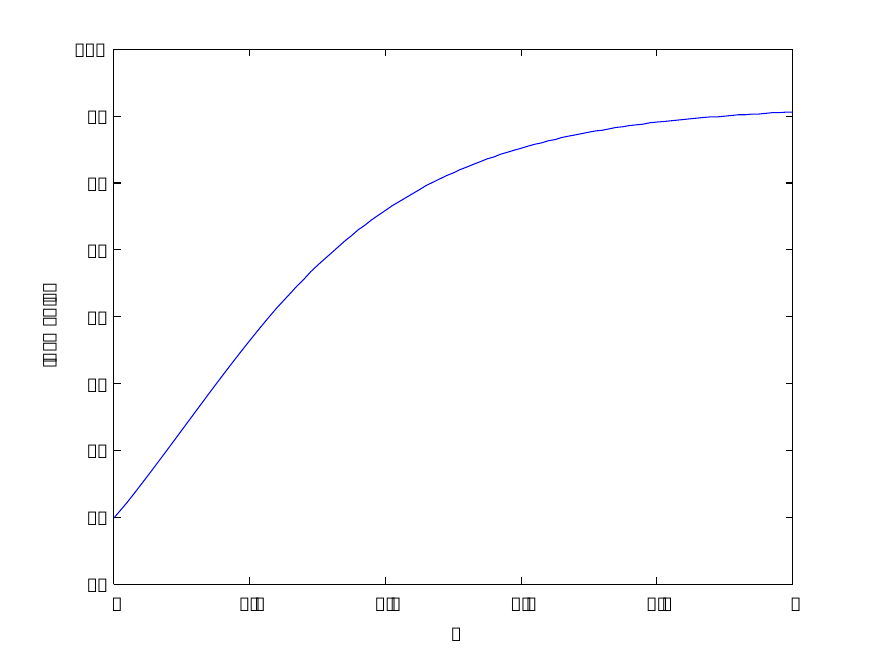}  & \includegraphics[scale=0.4]{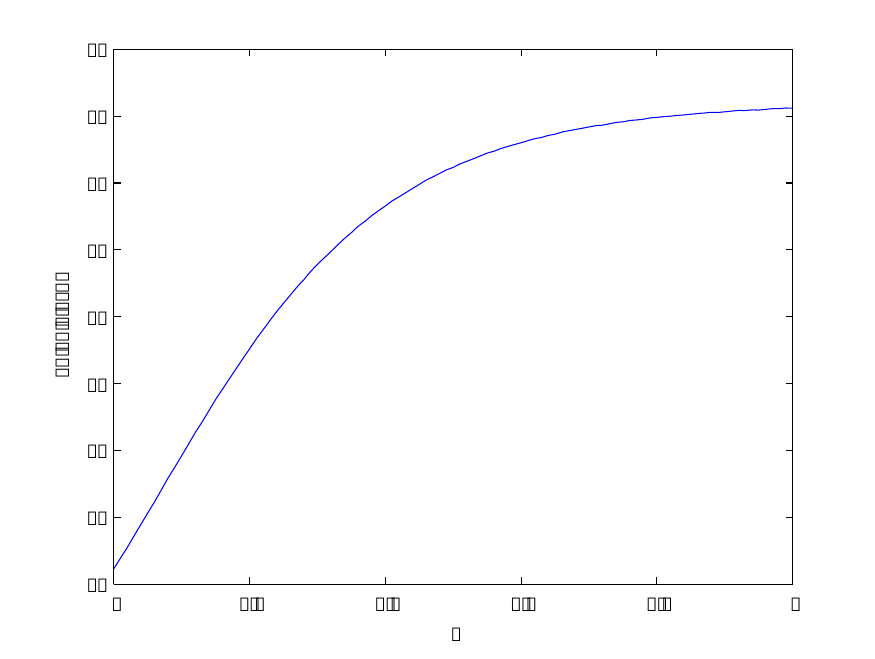} \\
(i) optimal face value $P^*(a)$ & (ii) bankruptcy level $e^{B^*}$  at $(a,P^*(a))$
\end{tabular}
\begin{tabular}{ccc}
\includegraphics[scale=0.4]{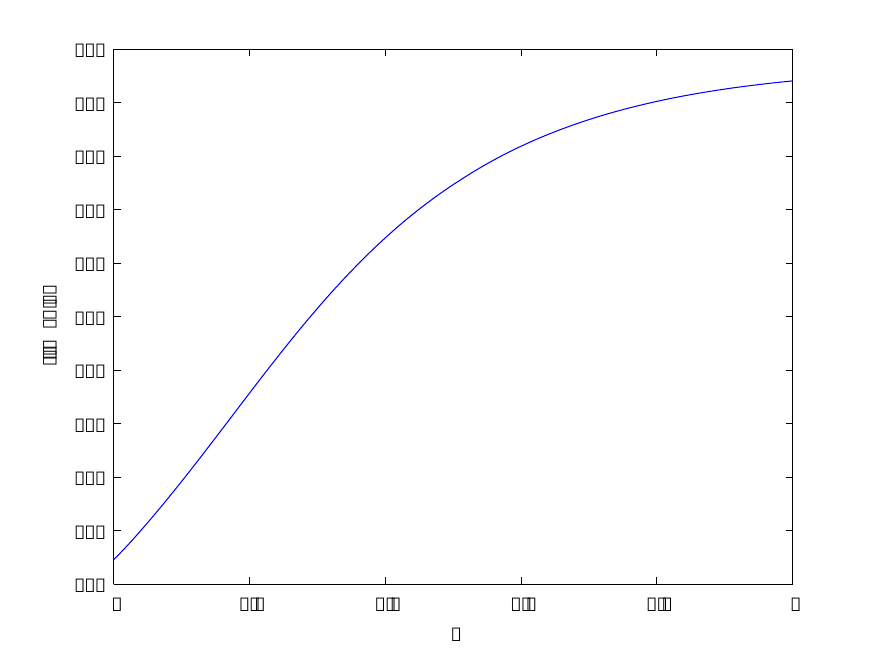}  & \includegraphics[scale=0.4]{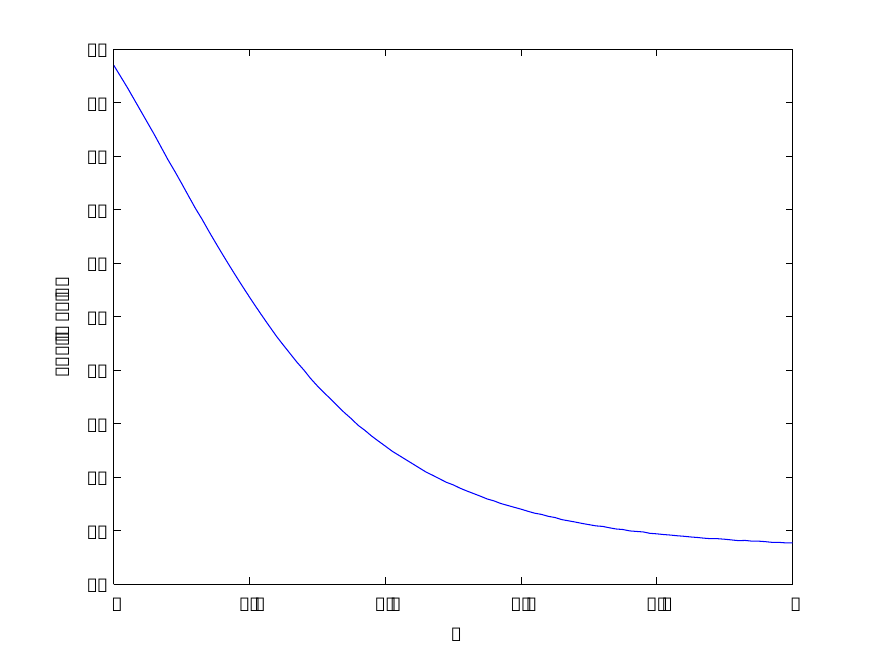}  & \includegraphics[scale=0.4]{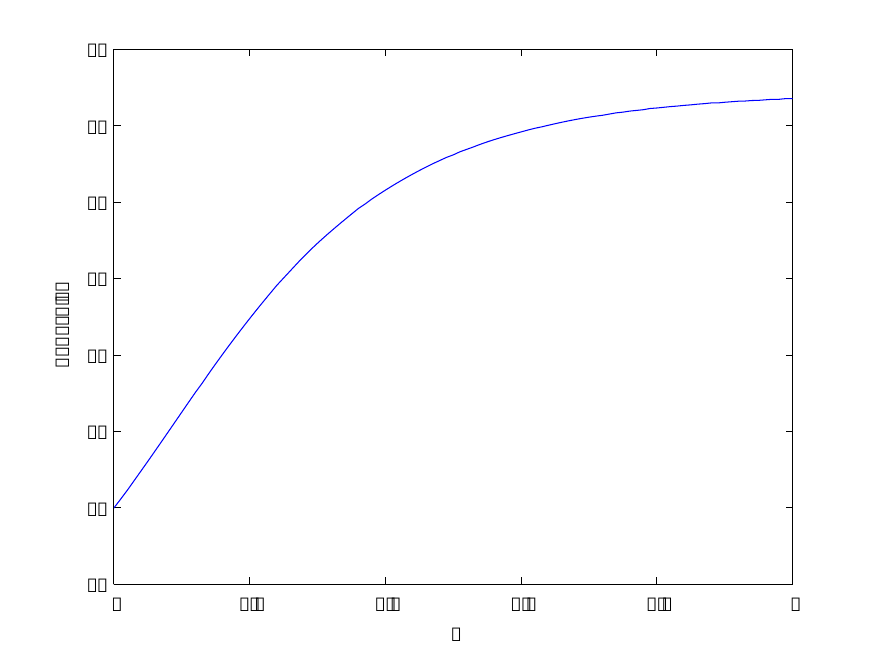} \\
(iii) firm value at $(a,P^*(a))$  & (iv) equity value  at $(a,P^*(a))$  & (v) debt value  at $(a,P^*(a))$ 
\end{tabular}
\end{minipage}
\caption{\small{Effects of bankruptcy cost concavity in the two-stage problem.}}  \label{two_stage_bankruptcy}
\end{center}
\end{figure}

\begin{figure}[htbp]
\begin{center}
\begin{minipage}{1.0\textwidth}
\centering
\begin{tabular}{cc}
\includegraphics[scale=0.4]{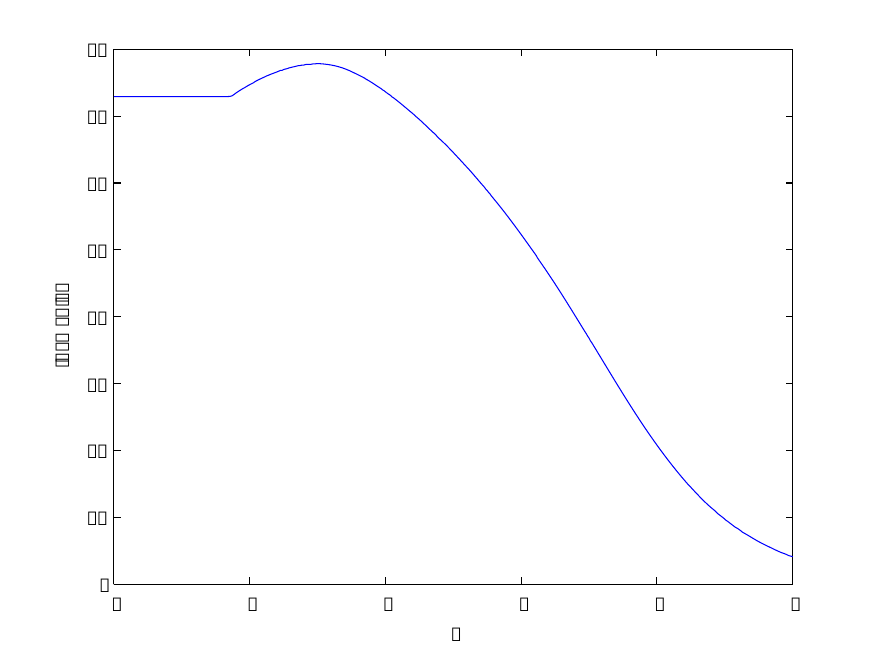}  & \includegraphics[scale=0.4]{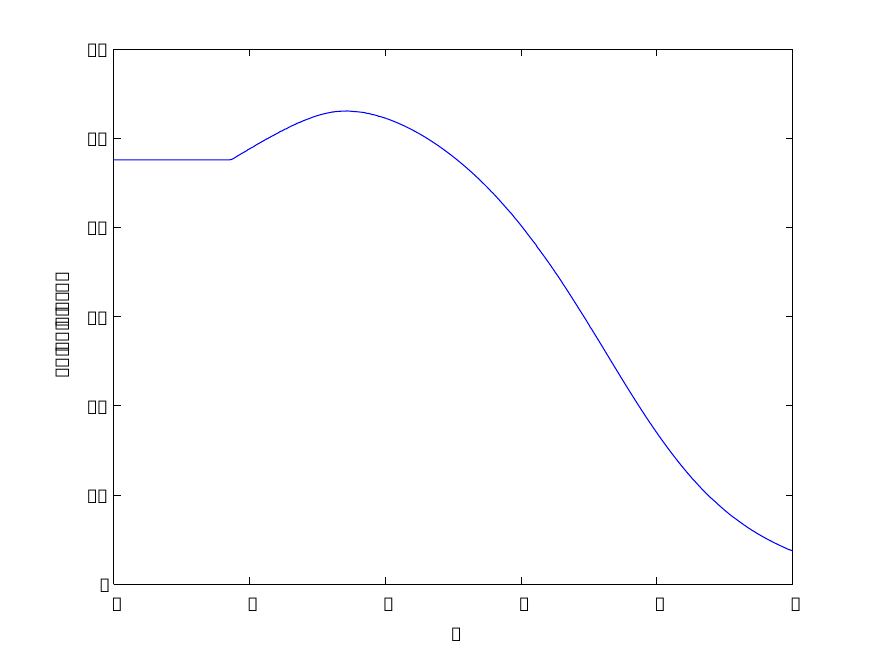} \\
(i) optimal face value $P^*(c)$ & (ii) bankruptcy level $e^{B^*}$  at $(c,P^*(c))$
\end{tabular}
\begin{tabular}{ccc}
\includegraphics[scale=0.4]{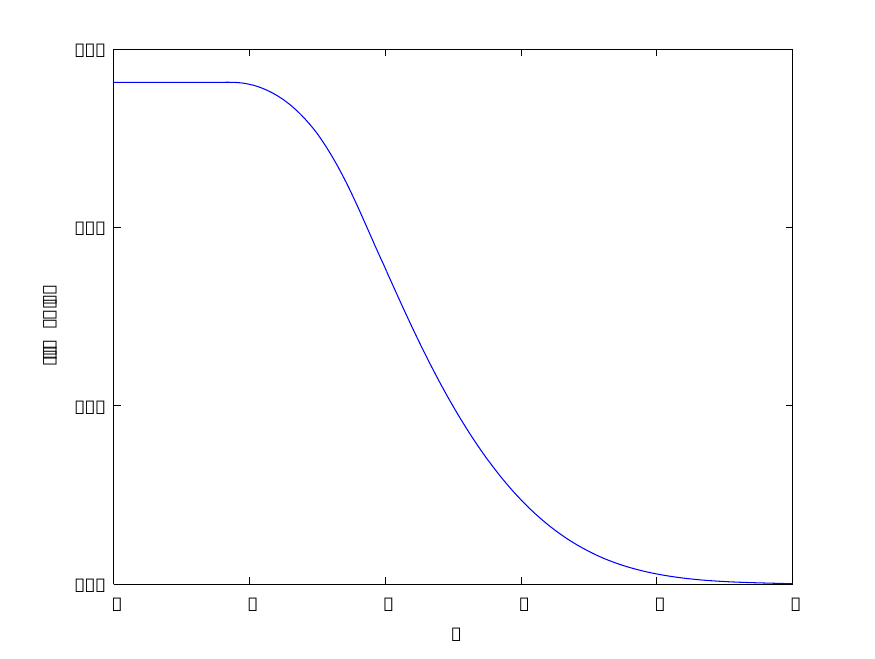}  & \includegraphics[scale=0.4]{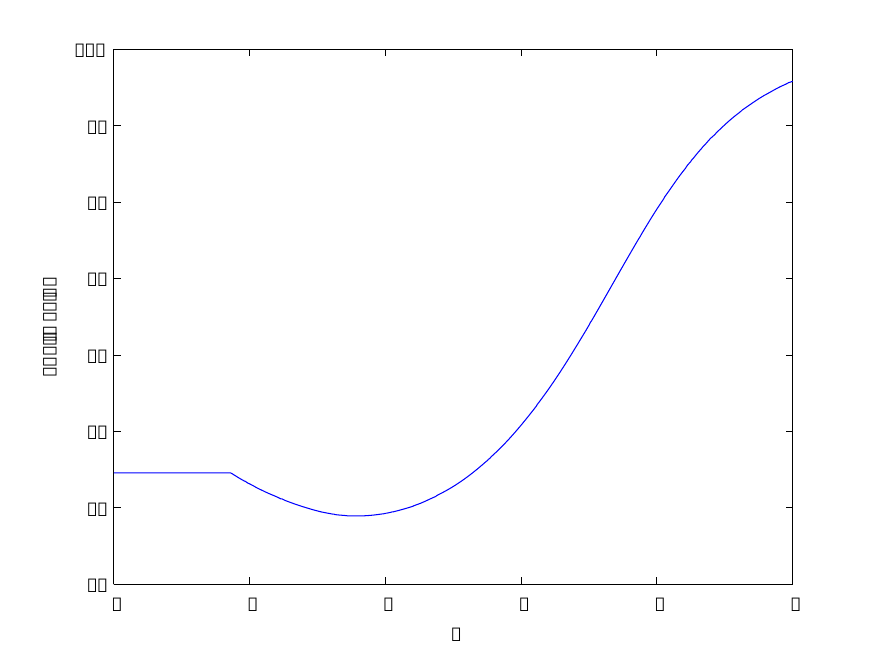}  & \includegraphics[scale=0.4]{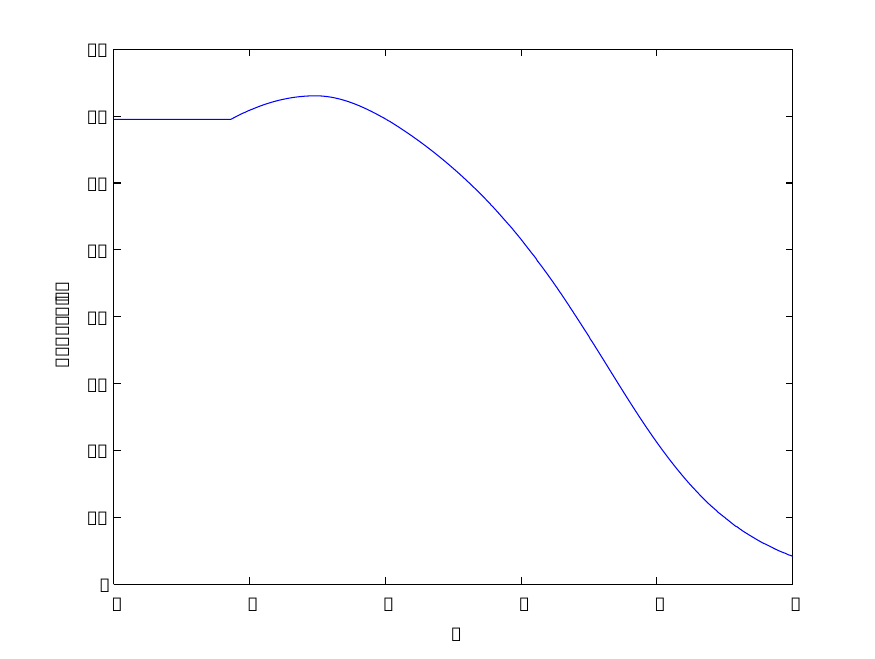} \\
(iii) firm value at $(c,P^*(c))$  & (iv) equity value  at $(c,P^*(c))$  & (v) debt value  at $(c,P^*(c))$ 
\end{tabular}
\end{minipage}
\caption{\small{Effects of tax convexity in the two-stage problem.}}  \label{two_stage_tax}
\end{center}
\end{figure}

We now move on to studying the impacts of the scale effects on the optimal capital structure.  We consider the two-stage problem \eqref{two_stage_problem} and let the face value $P$ be a control variable. 

Figure \ref{two_stage_bankruptcy} shows the impacts of the bankruptcy cost concavity on the optimal capital structure.  For each value of $a$ ranging from $0$ to $1$, we solve the two-stage problem \eqref{two_stage_problem} and compute (i) the optimal face value $P^*(a)$ as well as  (ii) bankruptcy level, and (iii)-(v) firm/equity/debt values corresponding to $(a, P^*(a))$.  From (i), we see that $P^*(a)$ is increasing.  This is because, as the bankruptcy cost decreases, the firm tends to rely more on debt financing  so as to enjoy more tax benefits.  While the problem is complicated due to the (conflicting) inner optimization where the equity holder chooses the bankruptcy level to maximize the equity value, it turns out that the (firm) value of the two-stage problem is monotonically decreasing in $a$ as seen in (iii).  
It can be confirmed that by choosing optimally the face value $P^*$, the firm value is necessarily no less than the unlevered asset value $V_0=100$ (which is attained by setting $P=0$).  Notice that this is not always true when the face value is fixed as in Figure  \ref{one_stage_bankruptcy}-(i).
The results in Figure \ref{two_stage_bankruptcy} are different from those in Figure \ref{one_stage_bankruptcy} because the face value also changes as $a$ changes.  Nonetheless, we still observe monotonicity and convexity/concavity in these figures.



Figure \ref{two_stage_tax} shows the impacts of tax convexity.  Similarly to Figure \ref{two_stage_bankruptcy}, for each value of $c$ ranging from $3$ to $8$, we compute the optimal face value $P^*(c)$ and the values corresponding to $(c,P^*(c))$.  Interestingly, the optimal face value $P^*(c)$ fails to be monotone and, for this reason, the bankruptcy level and equity/debt values also fail to be monotone.  Nevertheless,  the firm value, or the value of the two-stage problem,  is still monotonically decreasing.  The non-monotonicity of $P^*(c)$ is caused due to the nature of the tax benefits. As in the case of Figure  \ref{one_stage_tax}, the marginal effect of changing $c$ is negligible when $c$ is sufficiently large or sufficiently small.   For a medium value of $c$, we observe the monotonicity, which is consistent with the intuition that the tax convexity decreases the tax benefits and hence also the leverage.  However, for a neighborhood of $c$ where the marginal effect is very small but not exactly zero,  non-monotonicity of the optimal face value can happen due to the two conflicting maximizations of the equity value and the firm value. Also recall that the tax benefits vanish as $c$ increases.  We can confirm that, as $c$ increases, the optimal face value $P^*(c)$ converges to zero, thereby making the debt worthless, and the firm value converges to the unlevered asset value $V_0=100$.


%
%
%
%
%
%

In summary, scale effects have significant impacts on the capital structure.  Depending on the choice of bankruptcy cost and tax rate functions, the results can also become very complicated.
The computation we conducted here is efficient and does not rely on heavy algorithms.
 This is partly due to the form of its scale function \eqref{scale_function_exp} and also to the choice of $\eta$ and $f$ given above.  However, this can be extended very easily to the phase-type case; see \cite{Egami_Yamazaki_2010_2}.  For other spectrally negative \lev processes with explicit forms of scale functions, see \cite{Hubalek_Kyprianou_2009, Kyprianou_2008, Kyprianou_Surya_2007}.
We also remark that the solutions can in principle be computed numerically for any choice of spectrally negative \lev process by using the approximation algorithms of the scale function such as \cite{Egami_Yamazaki_2010_2,Surya_2008}.

%


\appendix

\section{Proofs}

\begin{proof}[Proof of Lemma \ref{lemma_Gamma_derivative_B}]
By  \eqref{def_Gamma} and \eqref{def_Theta}, the left-hand side equals
\begin{multline*}
e^B ( \Gamma^{(r+m)}(x-B)  - \Gamma^{(r+m)'}(x-B)  ) = - e^B \left[ \frac {\kappa(1)-(r+m)} {1 - \Phi(r+m)} W^{(r+m)'}(x-B)  \right. \\  \left. + (\kappa(1)-(r+m)) W^{(r+m)} (x-B) - \frac {\kappa(1)-(r+m)} {1 - \Phi(r+m)} W^{(r+m)}(x-B) \right],
\end{multline*}
which equals the right-hand side.
\end{proof}

\begin{proof}[Proof of Proposition \ref{prop_derivative_B}]
Applying Lemmas \ref{lemma_M_derivative_B}-\ref{lemma_Gamma_derivative_B} in \eqref{equity_in_terms_of_scale},
\begin{align*}
\frac {\partial} {\partial B}\mathcal{E} (x; B) = &- \Theta^{(r+m)}(x-B) \frac {\kappa(1)-(r+m)} {1 - \Phi(r+m)} e^B  \\ &+ \Theta^{(r+m)}(x-B) \Big[\frac {P \hat{\rho} + p} {\Phi(r+m)} + \frac {r+m} {\Phi(r+m)} \eta(B) + H^{(r+m)}(B) + \frac {\sigma^2} 2 \eta'(B)  \Big] \\
&-\Theta^{(r)}(x-B)  \Big[ G^{(r)}(B) + \frac r {\Phi(r)} \eta(B)  + H^{(r)}(B) + \frac {\sigma^2} 2 \eta'(B)  \Big]. 
\end{align*}
The claim holds by replacing the last term with
\begin{align*}
&- \Theta^{(r+m)}(x-B) \Big[ G^{(r)}(B) + \frac r {\Phi(r)} \eta(B)  + H^{(r)}(B) + \frac {\sigma^2} 2 \eta'(B)  \Big] \\
&- (\Theta^{(r)}(x-B) -  \Theta^{(r+m)}(x-B)) \Big[ G^{(r)}(B) + \frac r {\Phi(r)} \eta(B)  + H^{(r)}(B) + \frac {\sigma^2} 2 \eta'(B)  \Big].
\end{align*}
\end{proof}

\begin{proof}[Proof of Remark \ref{remark_j}]
Because, for any $q > 0$, $\frac q{\Phi(q)}= c + \frac 1 2 \sigma^2 \Phi(q) + \int_{(0,\infty)} \Pi(\diff u) \big( \frac {e^{-\Phi(q) u} - 1} {\Phi(q)} + u 1_{\{u \in (0,1)\}}\big)$,
\begin{align*}
\frac {r+m} {\Phi(r+m)} - \frac r {\Phi(r)}  &= \frac 1 2 \sigma^2 \left( \Phi(r+m) - \Phi(r) \right) - \int_{(0,\infty)} \Pi(\diff u) \left( \frac {e^{-\Phi(r) u} - 1} {\Phi(r)} + u 1_{\{u \in (0,1)\}}\right) \\ & \qquad + \int_{(0,\infty)} \Pi(\diff u) \left( \frac{e^{-\Phi(r+m) u} - 1} {\Phi(r+m)} +  u 1_{\{u \in (0,1)\}}\right) \\ &= \frac 1 2 \sigma^2 \left( \Phi(r+m) - \Phi(r) \right) + \int_{(0,\infty)} \Pi(\diff u) \left( \frac {1-e^{-\Phi(r) u} } {\Phi(r)} - \frac {1 - e^{-\Phi(r+m) u}} {\Phi(r+m)}\right). 
\end{align*}
%
Substituting this in \eqref{def_J}, we obtain the result.

\end{proof}

\begin{proof}[Proof of Lemma \ref{lemma_Theta_convergence}]
The result for the case $\sigma > 0$ is clear because, by \eqref{at_zero},  $W^{(q)}(0) = 0$ and $W^{(q)'}(0+) = 2 / \sigma^2$ for any $q > 0$.  Suppose $\sigma = 0$.  Then as in the proof of Lemma 4.4 of \cite{Kyprianou_Surya_2007}, we have
\begin{multline*}
\lim_{x \downarrow 0} \big[ W^{(r)'}(x) - W^{(r+m)'}(x) \big]  = \lim_{\lambda \uparrow \infty} \int_0^\infty \lambda e^{-\lambda x} \left[ W^{(r)'}(x) - W^{(r+m)'}(x) \right]  \diff x \\
= \lim_{\lambda \uparrow \infty} \left[ \frac {\lambda^2} {\kappa(\lambda) - r} - \frac {\lambda^2} {\kappa(\lambda) - (r+m)} \right] = -m \lim_{\lambda \uparrow \infty} \left[ \frac {\lambda} {\kappa(\lambda) - r}  \frac {\lambda} {\kappa(\lambda) - (r+m)} \right] = 0.
\end{multline*}
Here the last equality holds because for any $q > 0$, due to Fatou's lemma and $\int_{(0,1)} x \Pi(\diff x) = \infty$,
\begin{align*}
\frac {\kappa(\lambda) - q} \lambda = c 
 +\int_{( 0,\infty)}\left( \frac {e^{-\lambda
x}-1} \lambda + x 1_{\{0 <x<1\}} \right) \,\Pi(\diff x) - \frac q \lambda \xrightarrow{\lambda \uparrow \infty} \infty.
\end{align*}
This together with $W^{(r)}(0) = W^{(r+m)}(0) = 0$ shows the result.
\end{proof}

\begin{proof}[Proof of Lemma \ref{lemma_above_B_star}]  The proof is trivial when $B^* = -\infty$ and hence assume that $B^* \in (-\infty, \infty]$.

(i) Suppose $X$ is of bounded variation and fix $\widehat{B} < B^*$.  By \eqref{cond_K_1}, we have $K_1^{(r,m)}(\widehat{B}) < 0$.  But by \eqref{cont_fit_equation}, this implies $\mathcal{E}(\widehat{B}+; \widehat{B}) < 0$, violating \eqref{constraint}.  Therefore, those $B$ satisfying \eqref{constraint} must lie on $[B^*,\infty)$.

(ii) Suppose $X$ is of unbounded variation and again fix $\widehat{B} < B^*$.
For any sufficiently small $\delta > 0$ such that $K_1^{(r,m)}(x) < 0$ for every $\widehat{B} < x < \widehat{B} + \delta$, we have by  Proposition \ref{prop_derivative_B}
\begin{multline*}
\inf_{\widehat{B} \leq x \leq \widehat{B}+\delta, \widehat{B} \leq y \leq x} \left. \frac {\partial} {\partial B}\mathcal{E} (x; B) \right|_{B=y}\geq   \inf_{0 \leq y \leq \delta}\Theta^{(r+m)}(y) \inf_{\widehat{B} \leq B \leq \widehat{B} +\delta}|K_1^{(r,m)}(B)| \\ - \sup_{0 \leq y \leq \delta}\{\Theta^{(r)}(y) - \Theta^{(r+m)}(y) \} \sup_{\widehat{B} \leq B \leq \widehat{B}+\delta}|K_2^{(r)}(B)|.
\end{multline*}
This converges to some strictly positive value as $\delta \downarrow 0$ by Lemma \ref{lemma_Theta_convergence}, \eqref{at_zero} and \eqref{def_Theta}; namely,
\begin{align*}
\inf_{\widehat{B} \leq x \leq \widehat{B}+\delta_0, \widehat{B} \leq y \leq x}\left. \frac {\partial} {\partial B}\mathcal{E} (x; B) \right|_{B=y}  > 0,
\end{align*}
 for some $\delta_0 > 0$.
But by \eqref{cont_fit_equation}, $\mathcal{E}(\widehat{B}+\delta_0+; \widehat{B}+\delta_0) = 0$, implying $\mathcal{E}(\widehat{B}+\delta_0; \widehat{B}) < 0$, violating \eqref{constraint}.  
\end{proof}

\begin{proof}[Proof of Lemma \ref{lemma_small_j}]
Define, as the Laplace exponent of $X$ under $\p^1$ with the change of measure $ \left. \frac {\diff \p^1_0} {\diff \p_0}\right|_{\mathcal{F}_t} = e^{X_t - \kappa(1) t}$,
\begin{align*}
\kappa_1(\beta) :=\Big(  \sigma^2  + c - \int_{(0,1)} u (e^{-u}-1) \Pi(\diff u)\Big) \beta
+\frac{1}{2}\sigma^2 \beta^2 +\int_{( 0,\infty)}(e^{- \beta u}-1 + \beta u 1_{\{ u \in (0,1) \}}) e^{-u}\,\Pi(\diff u).
\end{align*}
Then, $\kappa_1(\Phi(q)-1) = \kappa(\Phi(q)) - \kappa(1) = q - \kappa(1)$ for any $q > 0$; see page 215 of \cite{Kyprianou_2006}.  This shows
\begin{align*}
\frac {\kappa(1)-r} {1-\Phi(r)} &= \frac {\kappa_1(\Phi(r)-1)} {\Phi(r)-1} 
= \sigma^2  + c - \int_{(0,1)} u (e^{-u}-1) \Pi(\diff u) \\
&+\frac{1}{2}\sigma^2 (\Phi(r)-1) +\int_{( 0,\infty)} \Big(\frac {e^{- (\Phi(r)-1) u}-1} {\Phi(r)-1} +  u 1_{\{ u \in (0,1) \}} \Big) e^{-u}\,\Pi(\diff u), \\
\frac {\kappa(1)-(r+m)} {1-\Phi(r+m)} &= \frac {\kappa_1(\Phi(r+m)-1)} {\Phi(r+m)-1} 
=  \sigma^2  + c - \int_{(0,1)} u (e^{-u}-1) \Pi(\diff u) \\
&+\frac{1}{2}\sigma^2 (\Phi(r+m)-1) +\int_{( 0,\infty)} \Big(\frac {e^{- (\Phi(r+m)-1) u}-1} {\Phi(r+m)-1} +  u 1_{\{ u \in (0,1) \}} \Big) e^{-u}\,\Pi(\diff u).
\end{align*}
Subtracting the former from the  latter, we obtain the result.
\end{proof}

\bibliographystyle{abbrv}
\bibliographystyle{apalike}

\bibliographystyle{agsm}
\bibliography{capital_structure_bib}

\begin{thebibliography}{10}

\bibitem{alili-kyp}
L.~Alili and A.~E. Kyprianou.
\newblock Some remarks on first passage of {L}\'{e}vy processes, the american
  put and smooth pasting.
\newblock {\em Ann. Appl. Probab.}, 15:2062--2080, 2004.

\bibitem{Ang_et_al_1982}
J.~Ang, J.~Chua, and J.~McConnell.
\newblock The administrative costs of corporate bankruptcy: A note.
\newblock {\em Journal of Finance}, 37(1):219--226, 1982.

\bibitem{avram-et-al-2004}
F.~Avram, A.~E. Kyprianou, and M.~R. Pistorius.
\newblock Exit problems for spectrally negative {L}\'{e}vy processes and
  applications to ({C}anadized) {R}ussion options.
\newblock {\em Ann. Appl. Probab.}, 14:215--235, 2004.

\bibitem{Avram_et_al_2007}
F.~Avram, Z.~Palmowski, and M.~R. Pistorius.
\newblock On the optimal dividend problem for a spectrally negative {L}\'evy
  process.
\newblock {\em Ann. Appl. Probab.}, 17(1):156--180, 2007.

\bibitem{Baurdoux2008}
E.~Baurdoux and A.~E. Kyprianou.
\newblock The {M}c{K}ean stochastic game driven by a spectrally negative
  {L}\'evy process.
\newblock {\em Electron. J. Probab.}, 13(8), 2008.

\bibitem{Baurdoux2009}
E.~Baurdoux and A.~E. Kyprianou.
\newblock The {S}hepp-{S}hiryaev stochastic game driven by a spectrally
  negative {L}\'evy process.
\newblock {\em Theory Probab. Appl.}, 53, 2009.

\bibitem{Bayraktar_2012}
E.~Bayraktar, A.~Kyprianou, and K.~Yamazaki.
\newblock On optimal dividends in the dual model.
\newblock {\em Astin Bull.}, 43(3), 2012.

\bibitem{Bayraktar_2013}
E.~Bayraktar, A.~Kyprianou, and K.~Yamazaki.
\newblock On optimal dividends in the dual model under transaction costs.
\newblock {\em Insur.\ Math.\ Econ.}, forthcoming.

\bibitem{Berens_1995}
C.~C. Berens, J.L.
\newblock The capital structure puzzle revisited.
\newblock {\em Rev. Finan. Stud.}, 8(4):1185--1208, 1995.

\bibitem{Bertoin_1996}
J.~Bertoin.
\newblock {\em L\'evy processes}, volume 121 of {\em Cambridge Tracts in
  Mathematics}.
\newblock Cambridge University Press, Cambridge, 1996.

\bibitem{Bris_2006}
A.~Bris, I.~Welch, and N.~Zhu.
\newblock The costs of bankruptcy: chapter 7 liquidation versus chapter 11
  reorganization.
\newblock {\em Journal of Finance}, 61(3):1253--1303, 2006.

\bibitem{Chan_2009}
T.~Chan, A.~Kyprianou, and M.~Savov.
\newblock Smoothness of scale functions for spectrally negative {L}\'evy
  processes.
\newblock {\em Probab. Theory Relat. Fields}, 150:691--708, 2011.

\bibitem{Chen_Kou_2009}
N.~Chen and S.~G. Kou.
\newblock Credit spreads, optimal capital structure, and implied volatility
  with endogenous default and jump risk.
\newblock {\em Math. Finance}, 19(3):343--378, 2009.

\bibitem{Dao_Jeanblanc_2012}
B.~Dao and M.~Jeanblanc.
\newblock Double-exponential jump-diffusion processes: a structural model of an
  endogenous default barrier with a rollover debt structure.
\newblock {\em Journal of Credit Risk}, 8(2):21--42, 2012.

\bibitem{Leung_Yamazaki_2011}
E.~Egami, T.~Leung, and K.~Yamazaki.
\newblock Default swap games driven by spectrally negative {L}\'{e}vy
  processes.
\newblock {\em Stochastic Process. Appl.}, 123(2):347--384, 2013.

\bibitem{Egami_Yamazaki_2010_2}
M.~Egami and K.~Yamazaki.
\newblock Phase-type fitting of scale functions for spectrally negative
  {L}\'evy processes.
\newblock {\em arXiv:1005.0064}, 2012.

\bibitem{Egami-Yamazaki-2010-1}
M.~Egami and K.~Yamazaki.
\newblock Precautional measures for credit risk management in jump models.
\newblock {\em Stochastics}, 85(1):111--143, 2013.

\bibitem{Egami-Yamazaki-2011}
M.~Egami and K.~Yamazaki.
\newblock On the continuous and smooth fit principle for optimal stopping
  problems in spectrally negative levy models.
\newblock {\em Adv. in Appl. Probab.}, forthcoming.

\bibitem{Francois_2004}
Fran\c{c}ois and Morellec.
\newblock Capital structure and asset prices: Some effects of bankruptcy
  procedures.
\newblock {\em Journal of Business}, 77(2):387--411, 2004.

\bibitem{Franks_Torous}
J.~R. Franks and W.~N. Torous.
\newblock An empirical investigation of {U}.{S}. firms in reorganization.
\newblock {\em Journal of Finance}, 44(3):747--769, 1989.

\bibitem{Graham_Smith_1999}
J.~Graham and C.~Smith.
\newblock Tax incentives to hedge.
\newblock {\em Journal of Finance}, 54(6):2241--2262, 1999.

\bibitem{Hilberink_Rogers_2002}
B.~Hilberink and L.~C.~G. Rogers.
\newblock Optimal capital structure and endogenous default.
\newblock {\em Finance Stoch.}, 6(2):237--263, 2002.

\bibitem{Hubalek_Kyprianou_2009}
F.~Hubalek and A.~E. Kyprianou.
\newblock Old and new examples of scale functions for spectrally negative \lev
  processes.
\newblock {\em Sixth Seminar on Stochastic Analysis, Random Fields and
  Applications, eds R. Dalang, M. Dozzi, F. Russo. Progress in Probability,
  Birkhäuser}, 63:119--145, 2011.

\bibitem{Kyprianou_2006}
A.~E. Kyprianou.
\newblock {\em Introductory lectures on fluctuations of {L}\'evy processes with
  applications}.
\newblock Universitext. Springer-Verlag, Berlin, 2006.

\bibitem{Kyprianou_Palmowski_2007}
A.~E. Kyprianou and Z.~Palmowski.
\newblock Distributional study of de {F}inetti's dividend problem for a general
  {L}\'evy insurance risk process.
\newblock {\em J. Appl. Probab.}, 44(2):428--448, 2007.

\bibitem{Kyprianou_2008}
A.~E. Kyprianou and V.~Rivero.
\newblock Special, conjugate and complete scale functions for spectrally
  negative {L}\'evy processes.
\newblock {\em Electron. J. Probab.}, 13:no. 57, 1672--1701, 2008.

\bibitem{Kyprianou_Surya_2007}
A.~E. Kyprianou and B.~A. Surya.
\newblock Principles of smooth and continuous fit in the determination of
  endogenous bankruptcy levels.
\newblock {\em Finance Stoch.}, 11(1):131--152, 2007.

\bibitem{Lambert_2000}
A.~Lambert.
\newblock Completely asymmetric {L}\'evy processes confined in a finite
  interval.
\newblock {\em Ann. Inst. H. Poincar\'e Probab. Statist.}, 36(2):251--274,
  2000.

\bibitem{Courtois_2008}
O.~Le~Courtois and F.~Quittard-Pinon.
\newblock The optimal capital structure of the firm with stable {L}\'evy assets
  returns.
\newblock {\em Decisions in Economics and Finance}, 31(1):51--72, 2008.

\bibitem{Leland_1994}
H.~E. Leland.
\newblock Corporate debt value, bond covenants, and optimal capital structure.
\newblock {\em J. Finance}, 49(4):1213--1252, 1994.

\bibitem{Leland_Toft_1996}
H.~E. Leland and K.~B. Toft.
\newblock Optimal capital structure, endogenous bankruptcy, and the term
  structure of credit spreads.
\newblock {\em J. Finance}, 51(3):987--1019, 1996.

\bibitem{Leung_Yamazaki_2010}
T.~Leung and K.~Yamazaki.
\newblock American step-up and step-down default swaps under {L}\'{e}vy models.
\newblock {\em Quant. Finance}, 14(1):137--157, 2013.

\bibitem{Loeffen_2008}
R.~L. Loeffen.
\newblock On optimality of the barrier strategy in de {F}inetti's dividend
  problem for spectrally negative {L}\'evy processes.
\newblock {\em Ann. Appl. Probab.}, 18(5):1669--1680, 2008.

\bibitem{Modilliani_Miller_1958}
F.~Modigliani and M.~Miller.
\newblock The cost of capital, corporation finance and the theory of
  investment.
\newblock {\em American Economic Review}, 48(3):261--297, 1958.

\bibitem{Modilliani_Miller_1963}
F.~Modigliani and M.~Miller.
\newblock Associationcorporate income taxes and the cost of capital: A
  correction.
\newblock {\em American Economic Review}, 53(3):433--443, 1963.

\bibitem{Sarkar_2008}
S.~Sarkar.
\newblock Can tax convexity be ignored in corporate financing decisions.
\newblock {\em Journal of Banking and Finance}, 32:1310--1321, 2008.

\bibitem{Surya_2008}
B.~A. Surya.
\newblock Evaluating scale functions of spectrally negative {L}\'evy processes.
\newblock {\em J. Appl. Probab.}, 45(1):135--149, 2008.

\bibitem{Thorburn_2000}
K.~Thorburn.
\newblock Bankruptcy auctions: Costs, debt recovery, and firm survival.
\newblock {\em Journal of Financial Economics}, 58(3):337--368, 2000.

\bibitem{Warner_1976}
J.~B. Warner.
\newblock Bankruptcy costs: Some evidence.
\newblock {\em Journal of Finance}, 32(2):337--347, 1976.

\bibitem{Yamazaki_2012}
K.~Yamazaki.
\newblock Contraction options and optimal multiple-stopping in spectrally
  negative {L}\'{e}vy models.
\newblock {\em arXiv:1209.1790}, 2012.

\end{thebibliography}

\end{document}